\documentclass[preprint,12pt]{elsarticle}

\usepackage{amssymb}
\usepackage{lineno}
\usepackage{subfigure}
\usepackage{amsmath}
\usepackage{amsthm}
\usepackage[linesnumbered,noend, ruled]{algorithm2e}
\usepackage{multirow}
\modulolinenumbers[5]

\newtheorem{definition}{Definition}
\newtheorem{lemma}{Lemma}
\newtheorem{theorem}{Theorem}

\begin{document}

\begin{frontmatter}



\title{Identifying Influential Users in Unknown Social Networks for Adaptive Incentive Allocation Under Budget Restriction}


\author[label1]{Shiqing Wu\corref{cor1}}
\author[label2]{Weihua Li}
%\author[label3]{Zike Zhang}
\author[label3]{Hao Shen}
\author[label1]{Quan Bai}

\affiliation[label1]{organization={University of Tasmania},%Department and Organization
            country={Australia}}
\affiliation[label2]{organization={Auckland University of Technology},%Department and Organization
            country={New Zealand}}
\affiliation[label3]{organization={fortiss GmbH, Forschungsinstitut des Freistaats Bayern},country={Germany}}
\cortext[cor1]{Corresponding Author. Email: shiqing.wu@utas.edu.au}

\begin{abstract}
In recent years, recommendation systems have been widely applied in many domains. These systems are impotent in affecting users to choose the behavior that the system expects. Meanwhile, providing incentives has been proven to be a more proactive way to affect users' behaviors. Due to the budget limitation, the number of users who can be incentivized is restricted. In this light, we intend to utilize social influence existing among users to enhance the effect of incentivization. Through incentivizing influential users directly, their followers or friends in the social network are possibly incentivized indirectly. However, in many real-world scenarios, the topological structure of the network is usually unknown, which makes identifying influential users difficult. To tackle the aforementioned challenges, in this paper, we propose a novel algorithm specifically for exploring influential users in unknown networks, which can estimate the influential relationships among users based on their historical behaviors and without knowing the topology of the network. Meanwhile, we design an adaptive incentive allocation approach that determines incentive values based on users' preferences and their influence ability. We evaluate the performance of the proposed approaches by conducting experiments on both synthetic and real-world datasets. The experimental results demonstrate the effectiveness of the proposed approaches.
\end{abstract}




\end{frontmatter}


\section{Introduction}\label{sec:introduction}
In recent years, recommendation systems have become increasingly popular and widely applied in many domains, such as e-commerce and entertainment \cite{Shani2011Evaluating, Bobadilla2013Recommender}. The objective of such systems is to passively satisfy users' explicit (e.g., quires) and implicit (e.g., potential interest) requirements via learning users' behavior patterns from a great volume of data beforehand. However, these systems can be impotent if some ``system-level'' objectives need to be proactively realized, such as persuading users to take a bus rather than drive for the commute. One effective solution is to provide incentives to users, as people tend to take actions that can produce return \cite{Homans1974Social}. Motivated by this background, how to effectively provide users with incentives attracts a lot of attention \cite{Singla2015Incentivizing, Gan2017Incentivize, Qiu2019Incentivizing}, which is computationally modeled as \emph{incentive allocation} problem. The goal of the incentive allocation problem is to maximize the number of users who take the behavior that the incentive provider expects by providing these users with incentives under a budget restriction \cite{Wu2021Learning}.

In general, an incentive provider tends to provide an incentive for each user as little as possible, maximizing the utility derived from the provided incentive. In the meantime, users expect they would receive an incentive as much as possible to maximize their individual profit as well. The ideal consequence is that all users can receive a desirable incentive and take the behavior that the incentive provider expects, such that both utilities of the incentive provider and users can be maximized. However, the conflict existing between the incentive provider and users makes incentive allocation challenging, since overpricing would waste the limited budget, while underpricing may lead to failure of incentivizing users \cite{Singla2013Truthful}. Hence, it is necessary to design a reasonable incentive structure and policies that consider the utilities of both the incentive provider and users. 

Meanwhile, due to the budget restriction, not all users could receive sufficient incentives. It implies that the overall result of user incentivization might be affected. To tackle this problem, some studies consider utilizing social influence to expand the effect of user incentivization \cite{Li2020Incentive, Zhang2020Incentivize}, as information diffusion can play an important role in propagating persuasive information among friends in a social network \cite{Li2018Influence}. Through such a process, a user's behaviors or decision-making can be influenced by the neighbors or other influential users in the same social network \cite{Axsen2013Social}. Namely, if we can appropriately incentivize some influential users in a social network, it is possible to affect users' behaviors indirectly via social influence. However, in the real world, recognizing influential users in a network beforehand is difficult, as the topology of a social network is often unknown in many scenarios \cite{Gjoka2010Walking,Maiya2011Benefits,Zhuang2013Influence}. Even though such data is available, it may still contain false or weak edges which are ineffective at spreading influence \cite{Wilder2018Maximizing, Bond2012A}, e.g., a user who has lots of fake followers in social media. Furthermore, using data harvested from different applications to analyze influential relationships among users could be not reliable, because towards distinct items or topics, a user may cause diverse influence on his or her adjacent users \cite{Tang2009Social}. For instance, a famous movie blogger might hardly affect his followers' decision on restaurants. Although the network can be reconstructed via surveys, and then existing algorithms applied to find out influential users \cite{Wilder2018Maximizing}, requiring lots of users involving makes such approaches impractical \cite{Valente2007Identifying}. In this case, an intelligent approach that can analyze users' influential ability and then discover reliable influential users in unknown social networks is essential. 

To achieve effective incentive allocation in an unknown social network, the following challenging issues have to be tackled. First, to identify influential users in unknown social networks, it is necessary to rely on information excluding the topology of the network for the analysis of users' influence ability. One potential candidate is the record of users' historical behaviors, as an influential relationship may exist between two users if one user's behaviors are always constant with that of the other user \cite{Yu2014Collective}. Nevertheless, the analysis based on the behavior record could be inaccurate, since it is common in a social network that two users perform the same behaviors even there is no connection between them \cite{Wu2019Adaptive}. In this case, it is difficult to identify the true influence relationship between each pair of users, but estimating such relationships can be possible. In addition, a reasonable incentive allocation policy is required to determine the incentives distributed to each user. One key factor in the incentive allocation policy is users' influence ability. As aforementioned, engaging influential users can help affect more users' behaviors under a limited budget. Hence, strong influential users suppose to be provided with more incentives for keeping them involved, while weak influential users or non-influential users may receive less even no incentives \cite{Wu2018GreenCommute}. Meanwhile, the overall status of users in the network needs to be considered to adjust incentives as well, since the marginal effect of allocating incentives would decline with the increasing number of users who take the behavior that the incentive provider expects \cite{Leskovec2007Cost}. Therefore, an adaptive incentive allocation is required for different situations in the social networks, e.g., the provided incentives to all users need to be decreased if most users have taken the target behavior, or increased conversely.

The recent work \cite{Wu2019Adaptive} conducted preliminary research work on modeling users' decision-making process under the impact of incentives and social influence and proposed the Agent-based Decision-making Model (ADM). In the ADM, all users are modeled as autonomous agents, and each user would make decisions based on individual utility consisting of preferences, offered incentives, and the social influence exerted by adjacent users. It also proposed a method to identify potential influential users in unknown social networks, named IPE. However, the performance of IPE could be limited since it only considers direct influence among users when identifying influential users. In this paper, we systematically elaborate on the incentive allocation problem in unknown social networks. We propose a novel Influential Users Discovery (IUD) algorithm, which is used to estimate the influential relationship between each pair of users based on their historical behaviors. Meanwhile, we propose the Dynamic Game-based Incentive Allocation (DGIA) algorithm for determining the value of incentives distributed to users. To evaluate the performance of the proposed approaches, extensive experiments on both synthetic and real-world datasets under the ADM model are conducted. The experimental results demonstrate that: (1) the IUD algorithm effectively estimates influential relationships among users and discovers influential users in unknown networks; (2) given the same budget and time span, the combination of IUD and DGIA algorithms outperforms other incentive allocation approaches. To summarize, the contributions of this research work are as follows:
\begin{itemize}

\item We propose an algorithm to estimate the influential relationship between each pair of users based on their historical behaviors, for identifying influential users in an unknown social network. Different from most existing approaches, the proposed algorithm can estimate influential relationships among users without knowing the topology of the network, such as edges among users, neighbors of a user, etc.
\item We propose an adaptive incentive allocation algorithm for determining the value of incentives, which considers users' influence ability and sensitivity to incentives, and the overall status of users in the network.
\item We use four real-world datasets as static social networks and three synthetic datasets as dynamic social networks to evaluate the performance of the proposed approaches. The experimental results demonstrate the effectiveness of the proposed approaches.  
\end{itemize}

The remainder of this paper is organized as follows. Section \ref{sec:related work} reviews the literature related to this study. Section \ref{sec:preliminaries} introduces the problem description and formal definitions. Subsequently, we first give details about the proposed approach for discovering potential influential users, then we explain the proposed approach for allocating incentives in Section \ref{sec:approach}. We demonstrate the experimental results for evaluating the performance of the proposed approach in Section \ref{sec:exp} and conclude this paper in Section \ref{sec:conclusions}.

\section{Related Work}\label{sec:related work}
\subsection{Incentive Allocation}
In recent years, many studies have been devoted to implementing effective incentive allocation for promoting users to take specific items or behaviors \cite{Qiu2019Incentivizing, Zhan2020Incentive, Wang2020Incentive, Lopez2020Cost}. Given a budget restriction, the goal of incentive allocation is to maximize the effect of users incentivization \cite{Wu2021Learning}. To better realize such a goal, some incentive allocation approaches are designed and tuned according to the features of scenarios, where the allocation policy determines incentive values based on specific attributes, such as users' preferences, location, and skill abilities \cite{Gan2017Incentivize, Wu2019Incentivizing, Li2019Redundancy, Qiu2019Incentivizing}. These approaches assume such necessary attributes are acquirable and utilize these attributes to model users' demands to incentives, such that effective incentive allocation can be realized. However, the performance of these attributes-based approaches could not be maintained when the data of necessary attributes is unavailable.

Meanwhile, inference-based approaches are also widely studied for incentive allocation. A typical example is the Budgeted Multi-Armed Bandit (BMAB) approaches, which model all incentive options as a set of arms and aim to find out the optimal arm to generate incentives \cite{Tran2010Epsilon, Tran2012Knapsack, Xia2015Thompson}. Singla et al. have proved the effectiveness of deploying BMAB in the real world \cite{Singla2015Incentivizing}. They proposed a UCB-based algorithm, named DBP-UCB, to incentivize users' engagement in helping to re-position sharing bikes among stations, avoiding the number of bikes in a station is excessive or lacking. However, most approaches assume that users' behaviors are independent, and they ignore that users' behaviors could be affected by external factors, such as social influence. This shortcoming makes these approaches fail to realize adaptive incentive allocation in social networks, where social influence changes frequently.

\subsection{Identifying Influential Users}
To better utilize the effect of social influence to affect users' behaviors, it is necessary to identify influential users in the social network and promote them first. The influence maximization problem, a similar problem aiming to select a number of users from a social network with the maximum influence spread, has been investigated over the past ten years \cite{Li2018Influence}. Although many studies have been conducted for solving the influence maximization problem from an algorithmic perspective \cite{Kempe2003Maximizing, Chen2009Efficient, Tang2014Influence, Nguyen2016Stop, Wang2021Maximizing, Calio2021Attribute, Li2021Social}, these works assume that the topology of a social network is explicitly given as input. However, knowledge about the network is usually unknown beforehand in many real-world applications and must be gathered via learning, observation, and survey \cite{Wu2021Learning, Wilder2018Maximizing}.

In this case, a few research works have been dedicated to identifying influential users in unknown social networks, where the topology of the network is not initially provided. For example, Mihara et al. proposed the IMUG algorithm that obtains true edges between users through limited probing and then determines the influential users \cite{Mihara2015Influence}. Wilder et al. proposed the ARISEN algorithm that queries random users to obtain true edges among users first and then applies the Influence Maximization algorithm to obtain the influential users \cite{Wilder2018Maximizing}. Eshghi et al. presented a novel approach inspired by the IMM algorithm \cite{Tang2014Influence} to tackle the influence maximization under the scenario where the network is partially visible \cite{Eshghi2019Efficient}. Kamarthi et al. designed a deep reinforcement learning-based method for learning to discover the unknown network structure for influence maximization \cite{Kamarthi2020Influence}. Nevertheless, the assumption applied to these aforementioned research works that users would return true information of the edges and topology of the network makes these approaches impractical. Meanwhile, these approaches only consider how to discover and reconstruct the network for activation of potential influential users to maximize influence diffusion, while they ignore the cost of incentivizing these users. This shortcoming makes these approaches unfeasible to be deployed for the incentive allocation problem constrained by the limited budget.

\subsection{Agent-based Modeling for Users' Behaviors}
Agent-Based Modeling (ABM) has been widely adopted in many areas for modeling complex systems, simulating continuous variations, and analyzing the trend of a particular phenomenon \cite{Bonabeau2002Agent, Macal2009Agent}. In the ABM, each entity is modeled as an autonomous agent that has the individual pattern for determining actions. This feature allows us to investigate the macro world from a micro level. For example, when modeling influence diffusion, some studies have abandoned popularly adopted influence diffusion models, e.g., Independent Cascade (IC) model and Linear Threshold (LT) model \cite{Kempe2003Maximizing}, since such models oversimplify the influence diffusion process and ignore that users' attributes and behaviors may affect the influence acceptance \cite{Li2019Automated}. In this case, they considered leveraging the advantages of the ABM to model the interactions among users from the user-centered perspective \cite{Li2019Automated, Jiang2015Diffusion, Li2018Modelling}. However, to the best of our knowledge, only a few studies adopted ABM to model the process of incentive allocation. Wu et al. proposed the Agent-based Decision-making (ADM) model that leverages ABM to model users' behaviors in the scenario where social influence and incentive co-exist simultaneously \cite{Wu2019Adaptive}. They also proposed the IPE method to identify potential influential users in unknown social networks. However, the performance of IPE could be limited since it only considers direct influence among users when identifying influential users. Different from \cite{Wu2019Adaptive}, in this paper, we consider both direct and indirect influence among users to estimate users' influential ability. Meanwhile, we will investigate how to tackle the incentive allocation problem in unknown networks based on the ADM model.

\section{Preliminaries}\label{sec:preliminaries}
\subsection{Problem Description}
\begin{figure}[t]
\centering
\includegraphics[width=\textwidth]{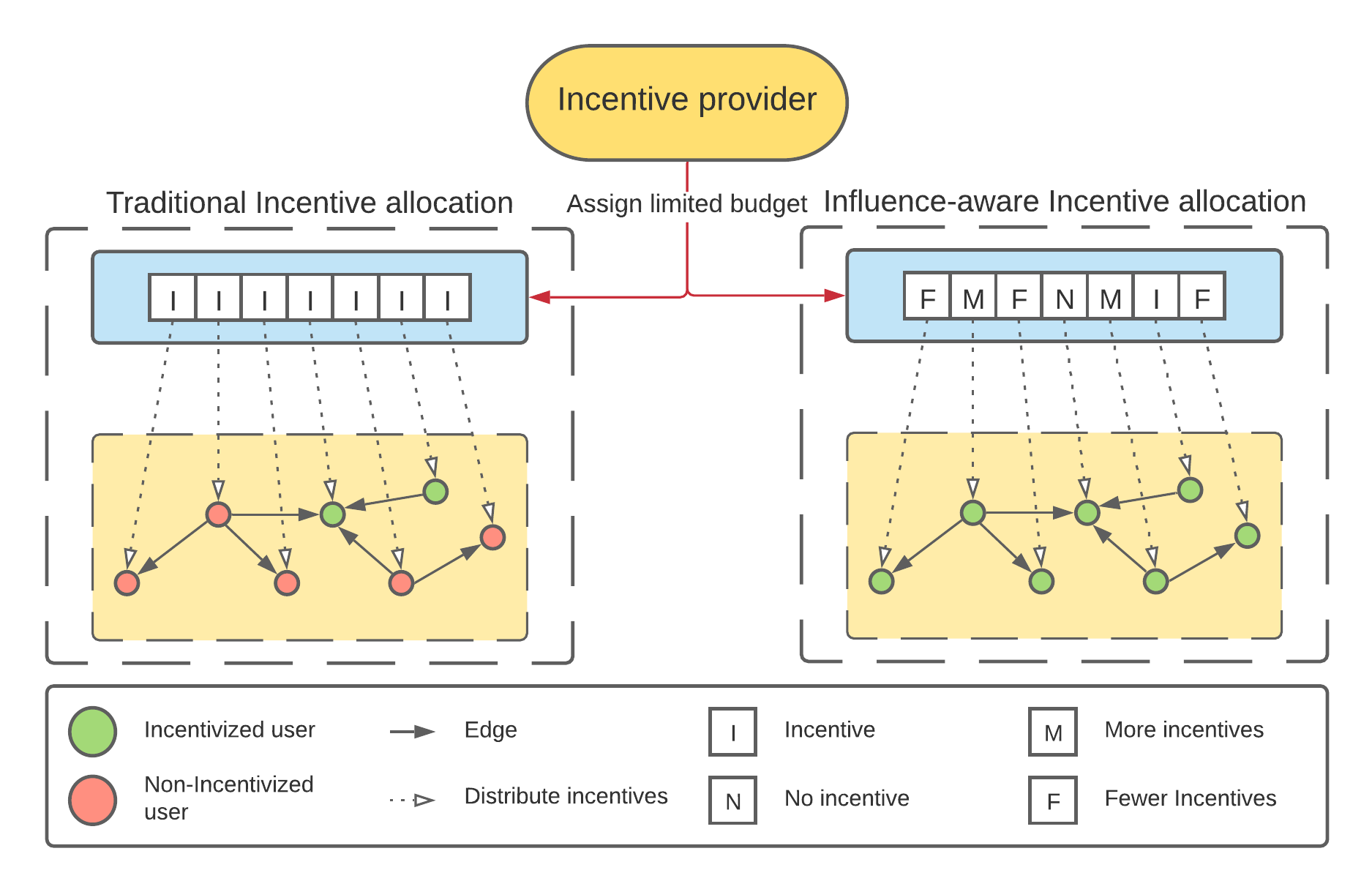}
\caption{Comparison of two different incentive allocation}
\label{fig:two_models}
\end{figure}

Preference is a key factor in people's decision-making process \cite{Allingham2002Choice}. However, people may start considering the overall utility of each available option before making decisions when additional factors they concentrate on occur \cite{Coleman1992Rational}. For example, many people may select soft drinks for purchase based on their flavor, while they may change their mind if there is an attractive promotion, e.g., discount. Here, the promotion can be regarded as an incentive, which increases the utility of the promoted items. Therefore, providing incentives is regarded as an effective method to affect users' decision-making process.

Many incentive allocation approaches have been developed \cite{Singla2015Incentivizing, Zhang2020Incentivize}. However, most existing approaches determine incentive values based on attributes of users only while ignoring the influence existing among them \cite{Gan2017Incentivize, Qiu2019Incentivizing, Wu2019Incentivizing}. As shown in Figure \ref{fig:two_models}, existing traditional incentive allocation may distribute incentives to all users based on its pricing policy, which is either customized or uniform. These approaches can effectively incentivize most users when the budget is sufficient. However, the performance would decrease rapidly once the budget is limited or the number of users becomes too large, since most users are failed to be incentivized due to insufficient incentives. Furthermore, when a social network exists among users, some influential users who are not incentivized may propagate natural or negative influence to their followers, reducing the effect of incentives simultaneously. Eventually, a possible result is that the incentive provider fails to incentivize most users. By contrast, influence-aware incentive allocation, which considers social influence among users when determining incentives, may perform much better in such scenarios. The reason is that influence-aware incentive allocation tends to allocate more incentive to influential users, expecting they can be successfully incentivized. These influential users can positively influence their followers such that the remaining users can be incentivized as well, even those non-influential users only receive fewer incentives or no incentive. Meanwhile, influence-aware incentive allocation can help save the limited budget for further allocation. Therefore, it is necessary to incentivize influential users in priority in social networks when the budget is insufficient.

As aforementioned, the information about social networks is usually unknown beforehand. Namely, the topology of the network is not given as input, such that every user's influential ability is difficult to measure. To realize effective influence-aware incentive allocation, in this paper, we first propose a method to estimate each user's influential ability for identifying influential users. Subsequently, we propose a novel approach that is able to allocate effective incentives to users.

\subsection{Formal Definitions}\label{subsec:definition}
Before introducing the proposed approaches, we first give formal definitions and introduce notations used in this paper. 

\begin{definition}\label{def:agent}
\textbf{A user agent}, $v_i$, denotes an autonomous agent representing a particular user in a social network. At at time step $t$, each $v_i$' s behavior can be represented by using $s_{v_i,t}$, where $s_{v_i,t} \in S_{v_i}$ and $S_{v_i}$ represents the set of $v_i$'s historical behaviors. The definition of possible actions that $v_i$ can take is introduced in Definition \ref{def:action}.
\end{definition}

\begin{definition}\label{def:network}
\textbf{A social network}, $G=(V,E)$, is defined as a directed graph, where $V=\{v_1,...,v_i\}$ denotes a set of user agents, and $E=\{e_{ij}|\{v_i,v_j\} \subseteq V\}$ denotes a set of edges in the network. Each directed edge $e_{ij}$ represents the influence from $v_i$ to $v_j$, and the weight $w_{ij}\in(0,1]$ associated with $e_{ij}$ represents the strength of the influence. In $G$, let $N_{v_i}^{out}$ be the set of agents who are directly influenced by $v_i$, i.e., $v_i$'s outgoing neighbors. Similarly, let $N_{v_i}^{in}$ be the set of agents who can influence $v_i$ directly, i.e., $v_i$'s incoming neighbors. Note that the affiliation information is maintained by each agent locally, and the system has no knowledge about the topology of $G$.
\end{definition}

\begin{definition}\label{def:action}
\textbf{An action option}, $a_m$, denotes a specific action that $v_i$ can choose, where $a_m\in A$ and $A=\{a_1,...,a_m\}$ represents a set of finite action options. For each $a_m$, $v_i$'s personal preference is defined as a fixed value $p_{v_i,a_m}\in[0,1]$, and the user utility at time step $t$ is represented as $u_{v_i,a_m,t}$. The action that the incentive provider expects users to choose is defined as $a^* \in A$. We regard a user is incentivized at time step $t$ if $s_{v_i,t}=a^*$.
\end{definition}

\begin{definition}\label{def:incentive}
\textbf{An incentive}, $r_{v_i,t}$, denotes the incentive that the system provider allocates to user $v_i$ if $v_i$ selects $a^*$ at time step $t$. The value of $r_{v_i,t}$ is constrained by the remaining budget $B_t$, i.e., $r_{v_i,t} \leq B_t$. Note that $v_i$ can only obtain the incentive when taking $a^*$, and after that, the correspoding value would be deducted from $B_t$. 
\end{definition}

\begin{definition}\label{def:influential_degree}
\textbf{Influential degree}, $\theta_{v_i}$, denotes $v_i$'s influential ability on affecting the rest of users over the whole network. The range of $\theta_{v_i}$ belongs to $[0,1]$, where higher $\theta_{v_i}$ implies that $v_i$ is potential to influence more users in the network, whereas smaller $\theta_{v_i}$ implies conversely. Note that $\theta_{v_i}$ is an estimated value, and the detail about calculation of $\theta_{v_i}$ would be introduced in Subsection \ref{subsec:iud}.
\end{definition}

\begin{definition}\label{def:incentive_sensitivity}
\textbf{Incentive sensitivity}, $\rho_{v_i,t}$, represents the sensitivity that $v_i$'s decision-making towards incentive. The range of $\rho_{v_i,t}$ is [0,1], where higher $\rho_{v_i,t}$ implies that $v_i$ requires fewer incentives to adopt $a^*$, and lower $\rho_{v_i,t}$ implies that incentivizing $v_i$ may need to spend more budget. Note that $\rho_{v_i,t}$ is an estimated value, and the detail about calculation of $\rho_{v_i,t}$ would be introduced in Subsection \ref{subsec:dgia}.
\end{definition}

\begin{definition}\label{def:GAUP}
\textbf{Global Activated Users Percentage (GAUP)} $\mu_t$, reflects the percentage of users who take $a^*$ at time step $t$ in the network. Higher $\mu_t$ implies more users select $a^*$ at time step  $t$, and lower $\mu_t$ implies conversely. $\mu_t$ can be formulated by using Equation \ref{equ:mu}, where the numerator denotes the number of users who select $a^*$ at time step $t$, and the denominator denotes the number of users registered in the social network.
\end{definition}

\begin{equation}\label{equ:mu}
\mu_t = \frac{{\left|\{v_i|s_{v_i,t}=a^*, v_i \in V\}\right|}}{|V|}
\end{equation}

\iffalse
\begin{table}[!t]
	\centering
	\scalebox{0.95}{
	\resizebox{\textwidth}{!}{
	\begin{tabular}{ll|ll}
		\hline
		Notation & Description & Notation & Description \\
		\hline
		$a_m$ & An action option&
		$a^*$ & The system-expected action\\
		$A$ & Set of available actions&
		
		$t$ & Time step\\
		
		$v_i$ & A user agent&
		$p_{v_i,a_m}$ & $v_i$'s preference towards $a_m$\\
		$u_{v_i,a_m,t}$ & $v_i$' utility towards $a_m$ at $t$&
		$s_{v_i,t}$ & $v_i$'s behavior at $t$\\
		$S_{v_i}$ & Set of $v_i$'s historical behaviors&
		
		$G$ & A social network\\
		$V$ & Set of all user agents&
		$e_{ij}$ & Directed edge from $v_i$ to $v_j$\\
		$E$ & Set of all edges connecting users&
		$w_{ij}$ & Weight of $e_{ij}$\\
		$N_{v_i}^{in}$ & Set of agents who influence $v_i$&
		$N_{v_i}^{out}$ & Set of agents who are influenced by $v_i$\\
		$k_{v_i,a_m,t}$ & Overall influence from $N_{v_i}^{in}$ at $t$&
		
		$r_{v_i,t}$ & Incentive allocated to $v_i$ at $t$\\
		$B_t$ & Remaining budget at $t$&
		
		$\theta_{v_i}$ & $v_i$'s influential degree\\
		
		$\rho_{v_i,t}$ & incentive sensitivity of $v_i$&
		$\omega_{v_i}$ & A ratio of $p_{v_i,a^*}$ to all preferences\\
		
		$\mu_t$ & Global Activated Users Percentage&
		
		$|\cdot|$ & The number of elements in a set\\
		$P(\cdot)$ & Probability\\
		\hline
	\end{tabular}
	}
	}
	\caption{Frequently Used Notations}
	\label{tab:1}
\end{table}
\fi
%The notations used in this paper are summarized in Table \ref{tab:1}. 
Given a limited time span $t \in [0,T]$ and a finite budget $B_t$ at each time step, the major objective is to incentivize users in a social network as many as possible, i.e., maximize $u_t$. To realize this goal, the proposed approach needs to discover influential users in an unknown network, and allocate reasonable incentives to users under the restriction of the limited budget.

\subsection{Agent-based Decision-Making Model}\label{subsec:adm}
In this study, we deploy the Agent-based Decision-Making (ADM) model \cite{Wu2019Adaptive} to simulate users' behaviors under the simultaneous impact of incentives and social influence. The ADM model is a decentralized model that inherits the advantages of Agent-based Modeling, i.e., each user is modeled as an autonomous agent, which maintains its personal attributes and chooses the action with the highest user utility at every time step. At time step $t$, $v_i$'s behavior $s_{v_i,t}$ can be determined by using Equation \ref{equ:beh}, where $u_{v_i,a_m,t}$ denotes the $v_i$'s user utility towards $a_m$. 

\begin{equation}\label{equ:beh}
s_{v_i,t}=\mathop{\arg\max}_{a_m \in A} u_{v_i,a_m,t}
\end{equation}

Towards a specific action $a_m$, $v_i$'s user utility is determined by three factors, i.e., $v_i$'s personal preference, social influence from $v_i$'s neighbors, and the incentive allocated to $v_i$. Equation \ref{equ:uti} calculates the user utility of $a_m$ at time step $t$, where $p_{v_i,a_m}$ denotes $v_i$'s preference, $k_{v_i,a_m,t}$ denotes the social influence exerted from $N_{v_i}^{in}$, and $r_{v_i,t}$ denotes the incentive. Note that, if $v_i$ receives no incentive, then the calculation of $u_{v_i, a^*, t}$ is exactly the same as $u_{v_i,a_m,t}$, where $a_m\in A\backslash\{a^*\}$. 

\begin{equation}\label{equ:uti}
u_{v_i,a_m,t} = \left\{ 
\begin{array}{ll}
p_{v_i,a_m}+k_{v_i,a_m,t}+r_{v_i,t}&,a_m = a^*\\
p_{v_i,a_m}+k_{v_i,a_m,t}&, a_m \neq a^*\\
\end{array}
\right.
\end{equation}

To model the value of $k_{v_i,a_m,t}$, inspired by the Linear Threshold (LT) Model \cite{Granovetter1978Threshold}, we assume that $k_{v_i,a_m,t}$ is aggregated by influence from $v_i$'s incoming neighbors who select $a_m$ at time step $t-1$. $k_{v_i,a_m,t}$ can be formulated by using Equation \ref{equ:inf}, where $N_{v_i}^{in}$ denotes the set of $v_i$'s incoming neighbors, and $w_{ji}$ represents the influence strength that $v_j$ influences $v_i$. To control the range of $k_{v_i,a_m,t}$, it satisfies that $\sum_{v_j \in N_{v_i}^{in}}w_{ji}\leq 1, \forall v_i \in V$. The benefit of inheriting the features of LT model is that aggregated $k_{v_i,a_m,t}$ is compatible in aggregating with $p_{v_i,a_m}$ and $r_{v_i,t}$ to form $u_{v_i,a_m,t}$.

\begin{equation}\label{equ:inf}
k_{v_i,a_m,t}=\sum_{\substack{
				v_j\in N_{v_i}^{in}\\
				s_{v_j, t-1}=a_m}
				}{w_{ji}}
\end{equation}

On the other hand, an assumption is underlying in the LT model that $v_i$ would be influenced if the total influence strength from its neighbors exceeds its associated threshold, which is a fixed value ranging from 0 to 1. However, this assumption is not compatible with the ADM since users' decision-making process is affected by two more factors, i.e., personal preference and incentives. Namely, the condition that $v_i$ is successfully incentivized in the ADM might be more complicated than that in the LT model. 

To better understand how to effectively incentivize users in the ADM under the budget restriction, we first analyze the condition that an incentive needs to satisfy. Given current time step $t$ and a user $v_i$, let $a'=\mathop{\arg\max}_{a_m \in A}{u_{v_i,a_m, t}}$ be the action with the $v_i$'s highest preference, and $a''=\mathop{\arg\max}_{a_m \in A}{u_{v_i,a_m, t}}$ be the action with the $v_i$'s highest utility before allocating incentives at time step $t$. Then, we have:

\begin{lemma}\label{lemma:1}
Suppose $N_{v_i}^{in}=\emptyset$, $v_i$ can be incentivized only if $r_{v_i,t}$ is at least $p_{v_i,a'} - p_{v_i,a^*}$.
\end{lemma}
\begin{proof}
Since $N_{v_i}^{in}=\emptyset$, then we can regard $k_{v_i,a_m,t}=0, \forall a_m \in A$. Hence, $u_{v_i,a',t}$ becomes $p_{v_i,a'}$, and $u_{v_i,a_m,t}$ turns to $p_{v_i,a*} + r_{v_i,t}$. In this case, $v_i$ can be incentivized only if $r_{v_i,t}$ is at least $p_{v_i,a'} - p_{v_i,a^*}$, and the lemma is proofed.
\begin{equation*}\label{equ:lemma1}
\begin{aligned}
&u_{v_i,a',t} \geq u_{v_i,a^*,t}
\Rightarrow r_{v_i,t} + p_{v_i,a^*} \geq p_{v_i,a'}
\Rightarrow r_{v_i,t} \geq p_{v_i,a'} - p_{v_i,a^*}
\end{aligned}
\end{equation*}
\end{proof}

\begin{lemma}\label{lemma:2}
Suppose $N_{v_i}^{in} \neq \emptyset$. If $v_i$ is expected to be incentivized when $r_{v_i,t}<p_{v_i,a'} - p_{v_i,a^*}$, then it must satisfy that $k_{v_i,a^*,t}>k_{v_i,a',t}$.
\end{lemma}
\begin{proof}
When $v_i$ is incentivized, $r_{v_i,t}$ needs to make $u_{v_i,a^*,t}$ at least $u_{v_i,a'',t}$. If $a' \neq a''$, then $r_{v_i,t}$ satisfies:

\begin{equation*}\label{equ:lemma21}
\begin{aligned}
&u_{v_i,a^*,t}\geq u_{v_i,a'',t}\\ 
\Rightarrow &r_{v_i,t}+p_{v_i,a^*}+k_{v_i,a^*,t}\geq p_{v_i,a''}+k_{v_i,a'',t}\\
\Rightarrow &r_{v_i,t} \geq p_{v_i,a''}-p_{v_i,a^*}+k_{v_i,a'',t}-k_{v_i,a^*,t}
\end{aligned}
\end{equation*}
Let $r_{v_i,t}<p_{v_i,a'} - p_{v_i,a^*}$, then we have:
\begin{equation*}
\begin{aligned}
&p_{v_i,a'} - p_{v_i,a^*}>r_{v_i,t}\geq p_{v_i,a''}-p_{v_i,a^*}+k_{v_i,a'',t}-k_{v_i,a^*,t}\\
\Rightarrow &p_{v_i,a'} + k_{v_i,a^*,t} > p_{v_i,a''}+k_{v_i,a'',t}=u_{v_i,a'',t}>p_{v_i,a'} + k_{v_i,a',t}\\
\Rightarrow &k_{v_i,a^*,t}>k_{v_i,a',t}
\end{aligned}
\end{equation*}
Whereas if $a' = a''$, $r_{v_i,t}$ should be at least the difference between $u_{v_i,a^*,t}$ and $u_{v_i,a',t}$. Similar to the case when $a'\neq a''$, we have:

\begin{equation*}\label{equ:lemma22}
\begin{aligned}
&u_{v_i,a^*,t}\geq u_{v_i,a',t}\\ 
\Rightarrow &r_{v_i,t}+p_{v_i,a^*}+k_{v_i,a^*,t}\geq p_{v_i,a'}+k_{v_i,a',t}\\
\Rightarrow &p_{v_i,a'} - p_{v_i,a^*}>r_{v_i,t}\geq p_{v_i,a'}-p_{v_i,a^*}+k_{v_i,a',t}-k_{v_i,a^*,t}\\
\Rightarrow &k_{v_i,a^*,t}>k_{v_i,a',t}
\end{aligned}
\end{equation*}
Therefore, the lemma is proofed.
\end{proof}

\begin{lemma}\label{lemma:3}
Suppose $N_{v_i}^{in} \neq \emptyset$. If $v_i$ is incentivized when $r_{v_i,t}>p_{v_i,a'} - p_{v_i,a^*}$, then it must satisfy that $k_{v_i,a^*,t}<k_{v_i,a',t}$.
\end{lemma}
\begin{proof}
Similar to the proof in Lemma \ref{lemma:2}. If $a'=a''$, we have:
\begin{equation}
\begin{aligned}
&u_{v_i,a^*,t}\geq u_{v_i,a',t}\\
\Rightarrow &r_{v_i,t}\geq p_{v_i,a'}-p_{v_i,a^*}+k_{v_i,a',t}-k_{v_i,a^*,t}>p_{v_i,a'} - p_{v_i,a^*}\\
\Rightarrow &k_{v_i,a',t}>k_{v_i,a^*,t}
\end{aligned}
\end{equation}
Otherwise, if $a'\neq a''$, we have:
\begin{equation}
\begin{aligned}
&u_{v_i,a^*,t}\geq u_{v_i,a'',t}\\
\Rightarrow &r_{v_i,t}\geq p_{v_i,a''}-p_{v_i,a^*}+k_{v_i,a'',t}-k_{v_i,a^*,t}>p_{v_i,a'} - p_{v_i,a^*}\\
\Rightarrow &p_{v_i,a''}+k_{v_i,a'',t}-k_{v_i,a^*,t}>p_{v_i,a'}\\
\Rightarrow &p_{v_i,a''}+k_{v_i,a'',t}-k_{v_i,a^*,t}>p_{v_i,a'}+k_{v_i,a',t}-k_{v_i,a^*,t}>p_{v_i,a'}\\
\Rightarrow &k_{v_i,a',t}>k_{v_i,a^*,t}
\end{aligned}
\end{equation}
Therefore, this lemma is proofed.
\end{proof}

\begin{theorem}\label{theorem:1}
In the ADM, incentivizing influential users in social networks can help influence their neighbors' behaviors as well as save the limited budget.
\end{theorem}
\begin{proof}
According to Lemmas \ref{lemma:1}, \ref{lemma:2}, and \ref{lemma:3}, when the influence exerted from $v_i$'s incoming neighbors is supporting select $a^*$, the incentive provider is possibly to incentivize $v_i$ by providing incentives less than the value when social influence does not exist. 
\end{proof}

Theorem \ref{theorem:1} reveals that (1) incentivizing influential users to affect their neighbors can improve the efficiency of the budget use, and (2) an adaptive pricing mechanism is necessary for determining the value of incentives. In the next sections, we will introduce the proposed two approaches for discovering influential users as well as incentive allocation, respectively.

\section{Influence-aware Incentive Allocation Model}\label{sec:approach}

\subsection{The IUD Algorithm}\label{subsec:iud}
As aforementioned, in this study, we consider the incentive allocation problem in scenarios where the topology of a social network is unknown at all. To make the incentive allocation effective in such scenarios, we propose a heuristic algorithm, called Influential Users Discovery (IUD), to discover potential influential users for incentivizing them to positively affect their neighbors' behaviors.

In general, if a user's behavior is always consistent with one adjacent user, there is a high chance that the user is influenced by the neighbor \cite{Yu2014Collective}. While the latest behaviors often have a higher influence than that of past behaviors, and users are more possibly affected by the influencer's recent behaviors in the reality \cite{Li2019Automated}. Thus, the influence from influencer's historical behaviors still cannot be ignored. In this case, we first regard the probability that $v_i$'s behavior at time step $t'$ affects $v_j$'s behavior at time step $t$ satisfies the principle of natural decay \cite{Benevenuto2009Characterizing, Fang2013Predicting}. Equation \ref{equ:P1} describes the attenuation, where $s_{v_i,t'}$ and $s_{v_j,t}$ denote behaviors of $v_i$ and $v_j$, respectively, and $\lambda$ denotes the attenuation constant. The probability is maximum when $t'=t-1$, and starts to decrease over time. The probability tends to 0 when the difference between $t$ and $t'$ becomes huge, so that the influence of the distant behaviors can be ignored. Concurrently, we regard $s_{v_i,t'}$ has not influence exerted on $s_{v_j,t}$ if $s_{v_i,t'} \neq s_{v_j,t}$. 

\begin{equation}\label{equ:P1}
P(s_{v_j,t}| s_{v_i,t'}) =\left\{ 
\begin{array}{ll}
e^{-\lambda \cdot (t-t')}&, s_{v_i,t'} =s_{v_j,t}\\
0 &, s_{v_i,t'} \neq s_{v_j,t}\\
\end{array}
\right.
\end{equation}

The probability $P(s_{v_j,t}|v_i)$ that $v_i$ influences $v_j$'s behavior at time step $t$ can be derived from the weighted average of $P(s_{v_j,t}| s_{v_i,t'})$. Specifically, if $v_j$'s behavior is consistent with $v_i$'s historical behaviors, $P(s_{v_j,t}|v_i)$ would tend to 1. It implies that $s_{v_j,t}$ is very possibly affected by $v_i$. $P(s_{v_j,t}|v_i)$ can be formulated by using Equation \ref{equ:P2}, where $t'$ denotes a historical time step.

\begin{equation}\label{equ:P2}
P(s_{v_j,t}|v_i) = \frac{\sum_{\substack{
 								s_{v_j,t}= s_{v_i,t'}\\
 								s_{v_i,t'} \in S_{v_i}\\
 							}}{P(s_{v_j,t}| s_{v_i,t'})}}
 						{\sum_{\substack{
 								s_{v_i,t'} \in S_{v_i}\\
 							}}{P(s_{v_j,t}| s_{v_i,t'})}}
\end{equation}

According to the probabilities that $v_i$ affects $v_j$'s specific behaviors, the probability that $v_i$ affects $v_j$ can be obtained from the average of $P(s_{v_j,t}|v_i)$, calculated by using Equation \ref{equ:P3}, where $|S_{v_j}|$ denotes the number of transactions of $v_j$'s historical behavior records.
\begin{equation}\label{equ:P3}
P(v_j|v_i) = \frac{\sum_{s_{v_j,t} \in S_{v_j}\\}{P(s_{v_j,t}|v_i)}}{|S_{v_j}|}
\end{equation}

After estimating the probability that $v_i$ affects $v_j$, the influential degree $\theta_{v_i}$ of $v_i$ can be formulated by using Equation \ref{equ:theta}, where $|V|$ denotes the number of users in the network.
\begin{equation}\label{equ:theta}
\theta_{v_i} = \frac{\sum_{v_j\in V\backslash\{v_i\}}{P(v_j|v_i)}}{|V|-1}						
\end{equation}

\begin{algorithm}[t]
  \caption{The IUD Algorithm}
  \label{alg:1}
  \KwIn {The set of users $V$, time step $t$}
  \KwOut {All users' influential degree $\theta_{v_i}, \forall v_i\in V$}
  	\For {$v_i \in V$ at time step $t$}{
  		Waiting for $v_i$ taking action\;
  		\For {$v_j \in V \backslash \{v_i\}$}{
  			\For{$s_{v_i,t}\in S_{v_i}$, $t'\in[1,t]$}{ 			
  				\For{$s_{v_j,t'}\in S_{v_j}$, $t'\in[0,t-1]$}{
  					Compute $P(s_{v_i,t}| s_{v_j,t'})$ using Equation \ref{equ:P1};\\
  				}
  				Compute $P(s_{v_i,t}|v_j)$ using Equation \ref{equ:P2};\\
  			}
  			Compute $P(v_i|v_j)$ using Equation \ref{equ:P3};\\
  		}
  		
  	}
  	\textbf{In the end of time step $t$:}\\
  	\For{$v_i \in V$}{
  		Compute $\theta_{v_i}$ using Equation \ref{equ:theta};\\
  	}
\end{algorithm} 

The detailed process of the IUD is described in Algorithm \ref{alg:1}. Line 2 obtains $v_i$'s behavior at time step $t$. Lines 3-6 aim to compute the probability that $v_i$'s behavior at time step $t$ is influenced by $v_j$'s behavior at time step $t'$. Line 7 computes the probability that $v_j$ affects $v_i$'s behavior at time step t, and Line 8 calculates the probability that $v_j$ affects $v_i$. Then in the end of time step $t$, Lines 10-11 update all users' influential degree $\theta_{v_i}$ by using Equation \ref{equ:theta}. 

The computational complexity of this algorithm is mainly determined by the number of users in the network and the volume of users' historical behavior records. Hence, the computational complexity is $O(n^2)$.

\subsection{The DGIA Algorithm} \label{subsec:dgia}
Under a limited budget, an optimal incentive allocation is to incentivize more users by providing incentives as little as possible. However, due to the uncertainty of social influence, the minimum incentive for incentivizing the user can be variable and unknown. Hence, to effectively allocate incentives, we design a Dynamic Game-based Incentives Allocation (DGIA) algorithm, where the interaction between a user agent and the system is modeled as a dynamic game. The system analyzes users' incentive sensitivity based on their behaviors, and then adjusts the value of the incentive.

\begin{figure}[tb]
\centering
\includegraphics[width=\textwidth]{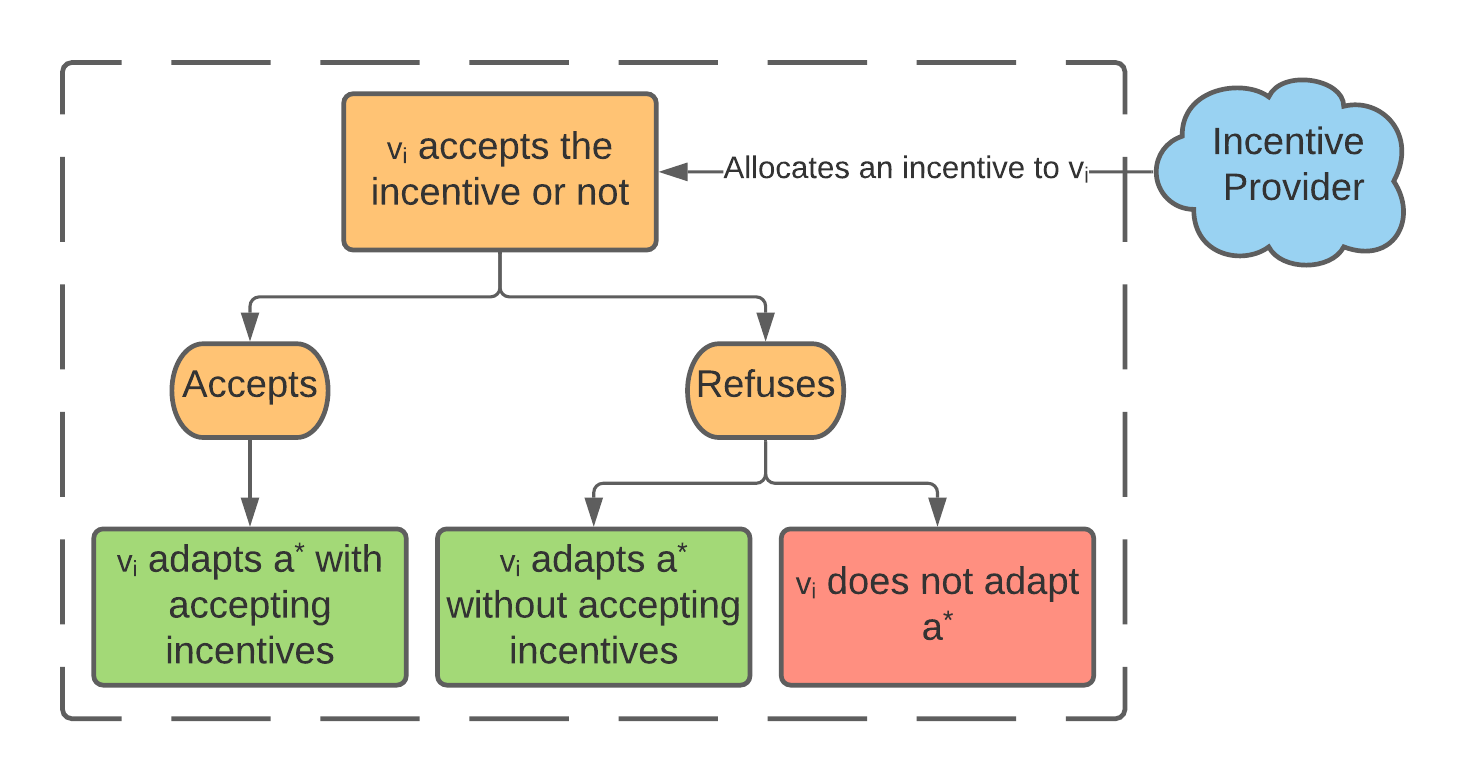}
\caption{$v_i$'s decision process}
\label{fig:dt}
\end{figure}

Standing on the perspective of the system, a user agent's decision process can be described by using Figure \ref{fig:dt}. The incentive provider provides an incentive to the user $v_i$ at the beginning for incentivizing $v_i$ to adopt $a^*$. Subsequently, $v_i$ makes a decision on whether accept the incentive or not based on its own utility. Note that, $v_i$ must adopt $a^*$ if the incentive is accepted. While, if $v_i$ refuses the incentive, i.e., $r_{v_i,t}$ in Equation \ref{equ:uti} is 0, then $v_i$'s utility is determined by only preferences and social influence. Due to the uncertainty of social influence in unknown social networks, we use Equation \ref{equ:omega} to formulate the probability that $v_i$ selects $a^*$ without accepting incentives, where $\omega_{v_i}$ is the ratio of $p_{v_i,a^*}$ to the sum of all $v_i$' preferences. If $p_{v_i,a^*}$ is dominant among preferences of the other $a_m$, the user is highly possible to select $a^*$ without any incentives, and otherwise conversely. Therefore, the probability that $v_i$ does not adopt $a^*$ is $1-\omega_{v_i}$. Meanwhile, we regard incentive sensitivity $\rho_{v_i,t-1}$ as the probability that $v_i$ accepts the incentive. 

\begin{equation}\label{equ:omega}
\omega_{v_i}=\frac{p_{v_i,a^*}}{\sum\limits_{a_m \in A} p_{v_i,a_m}}
\end{equation}

Based on $v_i$'s behavior, we can update $v_i$'s incentive sensitivity by using Equation \ref{equ:rho}, where $\gamma$ is a constant coefficient used to decrease incentive sensitivity. Here, we regard $\rho_{v_i,t-1}$ as a prior probability while $\rho_{v_i,t}$ as a posterior probability, representing the estimation of $v_i$'s incentive sensitivity. When $v_i$ selects $a^*$, $\rho_{v_i,t}$ will keep increasing and approaching to 1. Conversely, if $v_i$ is not incentivized, $\rho_{v_i,t}$ would decrease correspondingly. The value of $\gamma$ needs to be considered carefully, since too small $\gamma$ would waste budget for re-exploration, while too large $\gamma$ may cause the new incentive still fails to incentivize $v_i$.

\begin{equation}\label{equ:rho}
\rho_{v_i,t} =\left\{ 
\begin{array}{ll}
\frac{\rho_{v_i,t-1}}{\rho_{v_i,t-1}+\omega_{v_i}\cdot(1-\rho_{v_i,t-1})}&,s_{v_i,t} = a^*\\
\gamma\rho_{v_i,t-1}&,s_{v_i,t} \neq a^*\\
\end{array}
\right.
\end{equation}

The value of incentive $r_{v_i,t}$ can be calculated by using Equation \ref{equ:r}, where $\rho_{v_i,t-1}$ denotes $v_i$'s incentive sensitivity, $\theta_{v_i}$ denotes $v_i$'s influential degree, $\mu_{t-1}$ denotes the percentage of incentivied users at last time step, and $p_{v_i,a'}-p_{v_i,a^*}$ denotes the difference between $v_i$'s highest preference and the preference towards $a^*$. Overall, it can be seen that the value of $r_{v_i,t}$ is determined by following these rules: (1) The user who has larger dissatisfaction to $a^*$ may gain more incentives; (2) The user who has stronger influential degree may receive more incentives; (3) The user who is easier incentivized would receive fewer incentives; (4) Users would receive fewer incentives when the percentage of incentivized users is higher. 

\begin{equation}\label{equ:r}
r_{v_i,t}= (1-\rho_{v_i,{t-1}})\cdot((p_{v_i,a'}-p_{v_i,a^*})^{\mu_{t-1}}+(\theta_{v_i})^{\mu_{t-1}})
\end{equation}

\begin{algorithm}[t]
  \caption{The DGIA Algorithm}
  \label{alg:DGIA}
  \KwIn{The set of users $V$, budget $B$, time step $t$}
  \KwOut{Incentives for each user $r_{v_i,t}, \forall v_i \in V$}
  Initialize $B_t = B$\;
  	\For {$v_i \in V$ at time step $t$}{
  		Compute $r_{v_i,t}$ using Equation \ref{equ:r}\;
  		\If{$B_t<r_{v_i,t}$}{
  			$r_{v_i,t}=B_t$\;
  		}
		Allocate $r_{v_i,t}$ to $v_i$ and observe $s_{v_i,t}$\;
		\If{$s_{v_i,t} == a^*$}{
			$B_t = B_t - r_{v_i,t}$\;	
		}
		Compute $\rho_{v_i,t}$ using Equation \ref{equ:rho}\;
  	}
  	$sum:=0$\;
  	\For {$v_i \in V$ }{
  		Add $s_{v_i,t}$ into $S_{v_i}$\;
		\If{$s_{v_i,t}==a^*$}{
			$sum=sum+1$\;		
		}
  	}
  	$\mu_t:=sum/|V|$\;
  	Sort $V$ in descending order based on the sum of each user's influential degree and incentive sensitivity;\\
\end{algorithm}

The DGIA algorithm is displayed in Algorithm \ref{alg:DGIA}. Line 1 initializes the budget $B_t$ at time step $t$. Lines 3-5 calculate the value of the incentive, and line 6 allocates the incentive to $v_i$ and observes $v_i$'s behavior. According to $v_i$'s behavior, Lines 7-8 update remaining budget $B_t$ and Line 9 updates $v_i$'s incentive sensitivity. At the end of time step $t$, the GAUP would be calculated in Lines 10-15. Line 16 sorts all users in descending order based on users' influential degree and incentive sensitivity. It indicates that users who are influential and easily incentivized would have priority to be provided with incentives. As the worst-case time complexity of Algorithm \ref{alg:DGIA} is mainly determined by the loops started from Line 2 and Line 11, respectively. Therefore, the complexity is only $O(n)$. 

\iffalse
For convenient understanding, the ensemble approach of IUD and DGIA is described in Algorithm \ref{alg:ensemble}. 

\begin{algorithm}[t]
  \caption{IUD+DGIA}
  \small
  \label{alg:ensemble}
	\For{$t \in [1,T]$}{  
  		Initialize $B_t = B$\;
  		\For {$v_i \in V$ at time step $t$}{
			Allocate $r_{v_i,t}$ to $v_i$ and observe $s_{v_i,t}$\;
			Update remaining budget $B_t$;\\
			Update $\rho_{v_i,t}$ and calculate $P(v_i|v_j)$;
  		}
  		Update influential degree for all $v_i \in V$;\\
  		Calculate GAUP $\mu_t$;\\
  		
  	}

\end{algorithm}
\fi

\section{Experiments and Analysis}\label{sec:exp}
We conducted two major experiments for this research work. The first one aims to evaluate the performance of the proposed approaches in static social networks, i.e., the topology of the network is unchanged. In the second experiment, we further evaluate the proposed approaches in dynamic social networks, where users may join or quit the social network.
\subsection{Experimental Setup}
To evaluate the performance of the proposed approach, we deployed both synthetic and real-world datasets to construct social networks and conducted experiments. Before demonstrating the experimental results, we first illustrate the experimental setup, including the statistics of datasets, details of compared approaches, and initialization of parameters.
\subsubsection{Datasets}
In the experiments, the following real-world datasets are used. The statistics of these datasets are listed in Table \ref{tab:statistic}. The ACC at the last column denotes the average clustering coefficient which reflects the density of the network.

\begin{itemize}
\item \textbf{Facebook}\footnote{https://snap.stanford.edu/data/ego-Facebook.html} dataset includes 10 anonymized ego-networks, consisting of 4039 nodes and 88234 edges. The average number of edges is 21.8, and the average clustering coefficient is 0.606 \cite{Leskovec2012Learning}.

\item \textbf{Twitter}\footnote{https://snap.stanford.edu/data/ego-Twitter.html} dataset contains 973 networks, 81306 users and 1768149 edges. To diminish the computing time, in the experiments, we adopt one of the networks, which consists of 236 users and 2478 edges. The average number of edges is 10.5, and the average clustering coefficient is 0.303 \cite{Leskovec2012Learning}. 

\item \textbf{Wiki}\footnote{https://snap.stanford.edu/data/soc-sign-bitcoin-otc.html} dataset contains administrator elections and vote history data from 3 January 2008. There are 2794 elections with 103663 total votes and 7066 users participating in the elections. In the dataset, nodes refer to wikipedia users and edges represent votes from one user to another. The average number of edges is 14.6, and the average clustering coefficient is 0.14 \cite{Leskovec2010Signed}.

\item \textbf{Email}\footnote{https://snap.stanford.edu/data/email-Eu-core.html} dataset generates a network using email communication from a large European research insititution, which contains 1005 nodes and 25571 edges. In the dataset, nodes refer to users and edges indicate that at least one email was sent. The average number of edges is 25.4, and the average clustering coefficient is 0.399 \cite{Leskovec2007Graph}.

\item \textbf{Haverford}\footnote{http://networkrepository.com/socfb-Haverford76.php} dataset is a social network extracted from Facebook consisting of 1446 users with 59589 edges representing friendship ties. The average number of edges is 41.2, and the average clustering coefficient is 0.323 \cite{Traud2012Social}.
\end{itemize} 
\begin{table}[!tb]
	\centering
	\scalebox{1}{
	\begin{tabular}{lllll}
		\hline
		Dataset&$|V|$&$|E|$&Avg. $|E|$&ACC\\
		\hline
		Facebook&4039&88234&21.8&0.606\\
		Twitter&236&2478&10.5&0.303\\
		Wiki&7115&103689&14.6&0.14\\
		Email&1005&25571&25.4&0.399\\
		Socfb&1446&59589&41.2&0.323\\
		\hline
	\end{tabular}
	}
	\caption{Statistics of Real-World Social Network Datasets}
	\label{tab:statistic}
\end{table}

The former four datasets are used to build four static networks, i.e., the topology of each network is fixed and would not be changed. While the last one dataset is used to form a dynamic network, where users may join the network and establish connections with existing users or exit the network at every time step. The configuration of the dynamic network would be introduced in Section \ref{subsec:dynamic}. 

Meanwhile, pre-processing data is necessary since these datasets only contain the nodes and unweighted edges. First, we assign a random weight for each edge, which satisfies that the sum of the weight of edges that connect to a user cannot exceed 1, i.e., $\sum_{v_j\in N_{v_i}^in}{w_{ji}}\leq 1$. Subsequently, towards each action option, we also assign a random preference from 0 to 1 for each user.

\subsubsection{Baseline Approaches}
We compare the proposed IUD and DGIA approaches with several incentive allocation approaches as follows:
\begin{itemize}
\item \textbf{DGIA} is a part of the proposed approach, which allocates incentives without considering users' influential degree.
\item \textbf{IPE+DGIA} is an ensemble approach that applies the IPE approach to discover influential users, and DGIA approach to allocate incentives. Different from the proposed IUD approach, IPE analyzes users' influential degree only based on their most recent behaviors \cite{Wu2019Adaptive}.
\item \textbf{DBP-UCB} is a bandit-based approach for engaging users to participate in the bike re-positioning process \cite{Singla2015Incentivizing}. DBP-UCB is a dynamic pricing mechanism, which can determine the incentives from a finite price list. According to the features of ADM model, we set nine price options from 0 to 2 with an interval of 0.25, i.e., $\{0, 0.25, 0.5, 0.75, 1, 1.25, 1.5, 1.75, 2.0\}$.
\item \textbf{Uniform Allocation} tends to provide fixed and uniform incentives to all users. The value of incentives is determined by the budget and the number of users.
\item \textbf{No Incentive} implies that the system would not incentivize users' behaviors. All user agents would make decisions based on their own preferences and exerted influence from neighbors. 
\end{itemize}

\subsubsection{System Setup}
We simulate the users' decision-making process under incentives and social influence by creating a number of user agents based on the datasets mentioned above. Each user agent maintains its own local information, including a list of neighbors the agent can influence, a list of neighbors who can influence the agent, preferences towards all available action options, and historical behavior records. Meanwhile, we use a system agent to represent the system for allocating incentives and collecting global information. 

Without specific clarification, we set $|A|=4$ by default, representing there are four possible action options. The budget provided at each time step is fixed and limited and is different for different networks. Each user agent's incentive sensitivity and influential degree are initialized as 0.5 and 0, respectively. The attenuation constant $\lambda$ is set as 0.1, and the coefficient $\gamma$ is set as 0.9. Furthermore, we assume that all user agents would join the system at time step 0, and start selecting an action from possible options from time step 1.

\subsubsection{Evaluation Metrics}
To better evaluate the performance of the proposed approach, two major evaluation metrics are considered as follows:
\begin{itemize}
\item \emph{Global Activated Users Percentage (GAUP)} computes the percentage of users who have been incentivized in a network, which has been formulated in Equation \ref{equ:mu}. Higher GAUP indicates more users in the network are incentivized.

\item \emph{Global Influenced Activation Coverage (GIAC)} computes the percentage of users who are incentivized due to social influence. As explained in Section \ref{subsec:adm}, when the provided incentive is insufficient, a user can still be incentivized if the influence from incoming neighbors who select $a^*$ exceeds the preference gap. Namely, the user cannot be incentivized without the influence in such scenarios. We use $\tau_t$ to represent the GIAC at time step $t$, and $\tau_t$ can be formulated by using Equation \ref{equ:GIAC}, where $p_{v_i,a'}$ denotes $v_i$'s highest preference among $a_m\in A$, the numerator denotes the number of users who are incentivized due to social influence, and the denominator denotes the number of users in the network.
\begin{equation}\label{equ:GIAC}
\tau_t = \frac{|\{v_i|s_{v_i,t}=a^*, r_{v_i,t}<p_{v_i,a'}-p_{v_i,a^*}\}|}{|V|}
\end{equation}
\end{itemize}

\subsection{Experimental Results on Static Networks}
In this experiment, we evaluate the performance of the proposed approach against other baseline approaches by using four real-world social networks, i.e., Facebook, Twitter, Wiki, and Email. To make the comparison meaningful, we provide the same but limited amount of budget to all approaches in different networks. 

\subsubsection{Performance on GAUP}
\begin{figure*}[!t]
  \centering
  \subfigure[Facebook, $B_t=200$]{\includegraphics[width=0.49\textwidth]{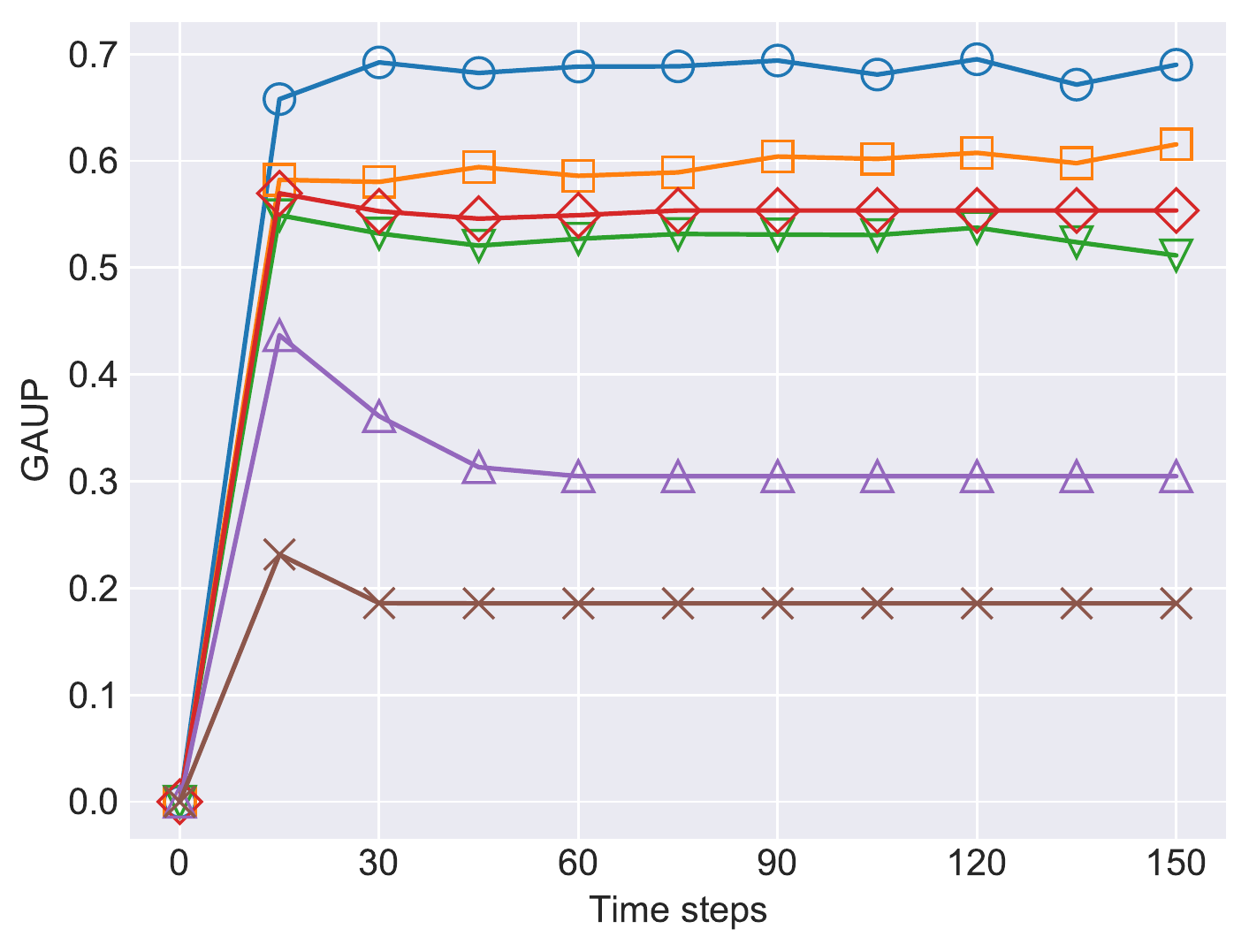}\label{fig:gaup-a}}
  \subfigure[Twitter, $B_t=20$]{\includegraphics[width=0.49\textwidth]{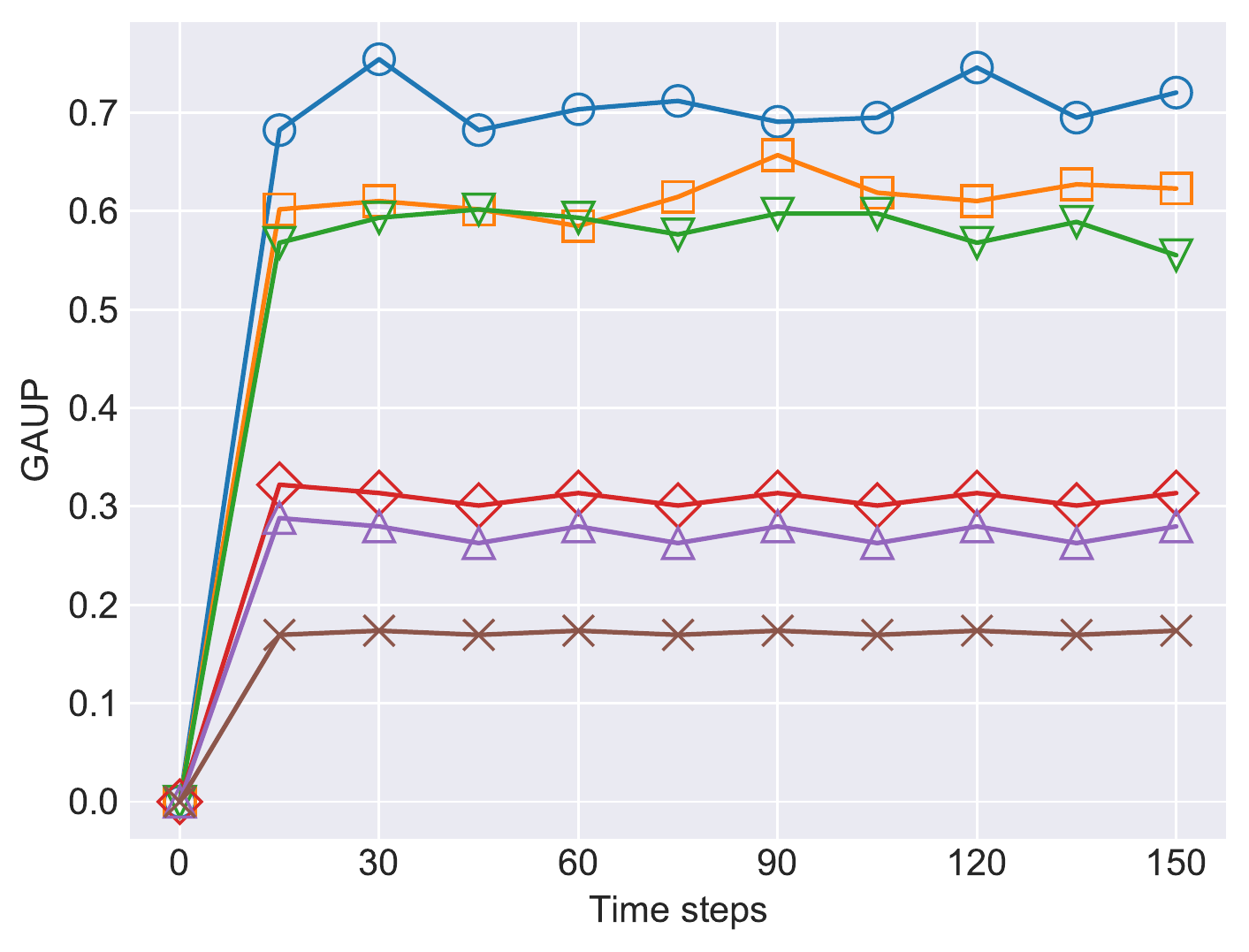}\label{fig:gaup-b}}\\
  \subfigure[Wiki, $B_t=700$]{\includegraphics[width=0.49\textwidth]{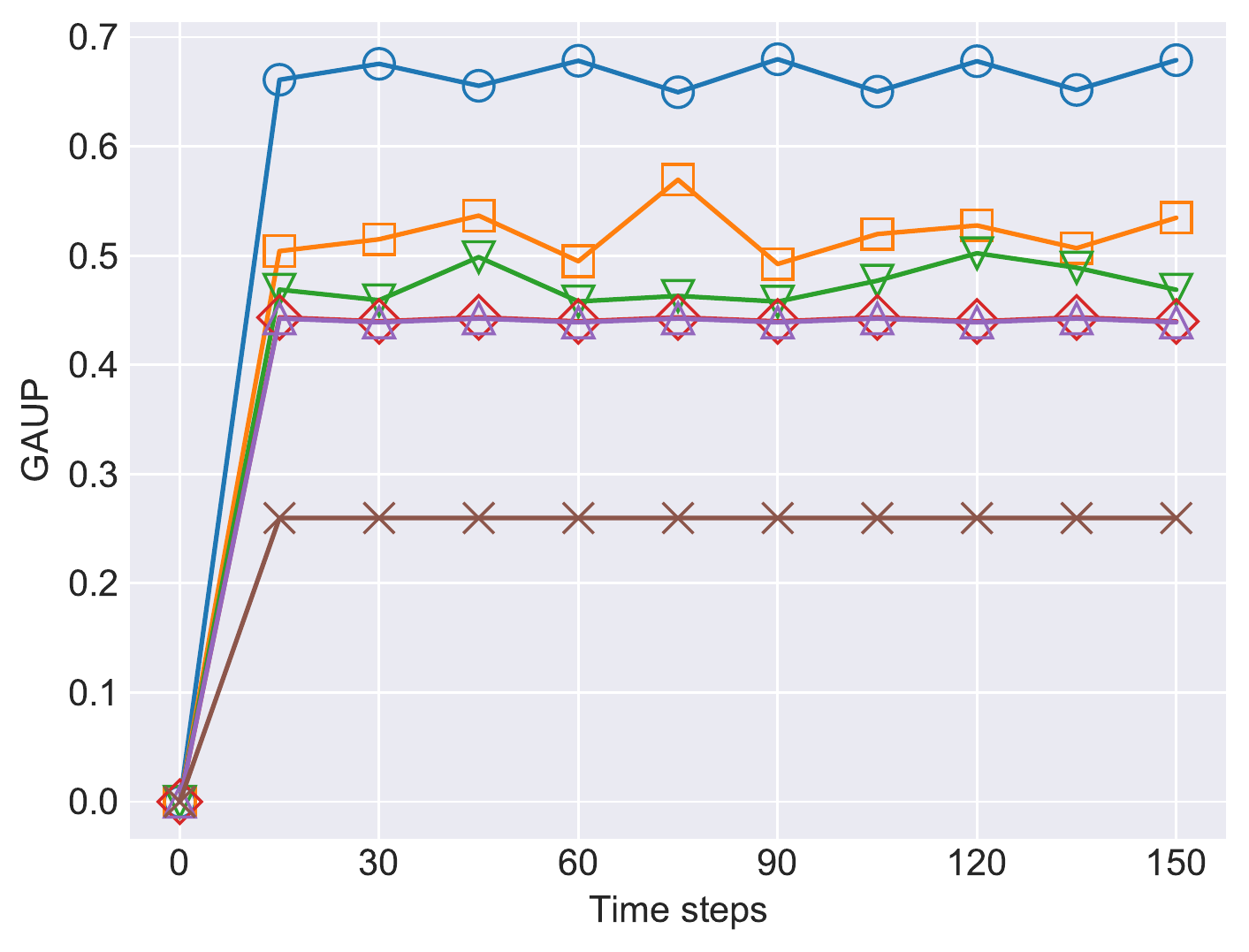}\label{fig:gaup-c}}
  \subfigure[Email, $B_t=50$]{\includegraphics[width=0.49\textwidth]{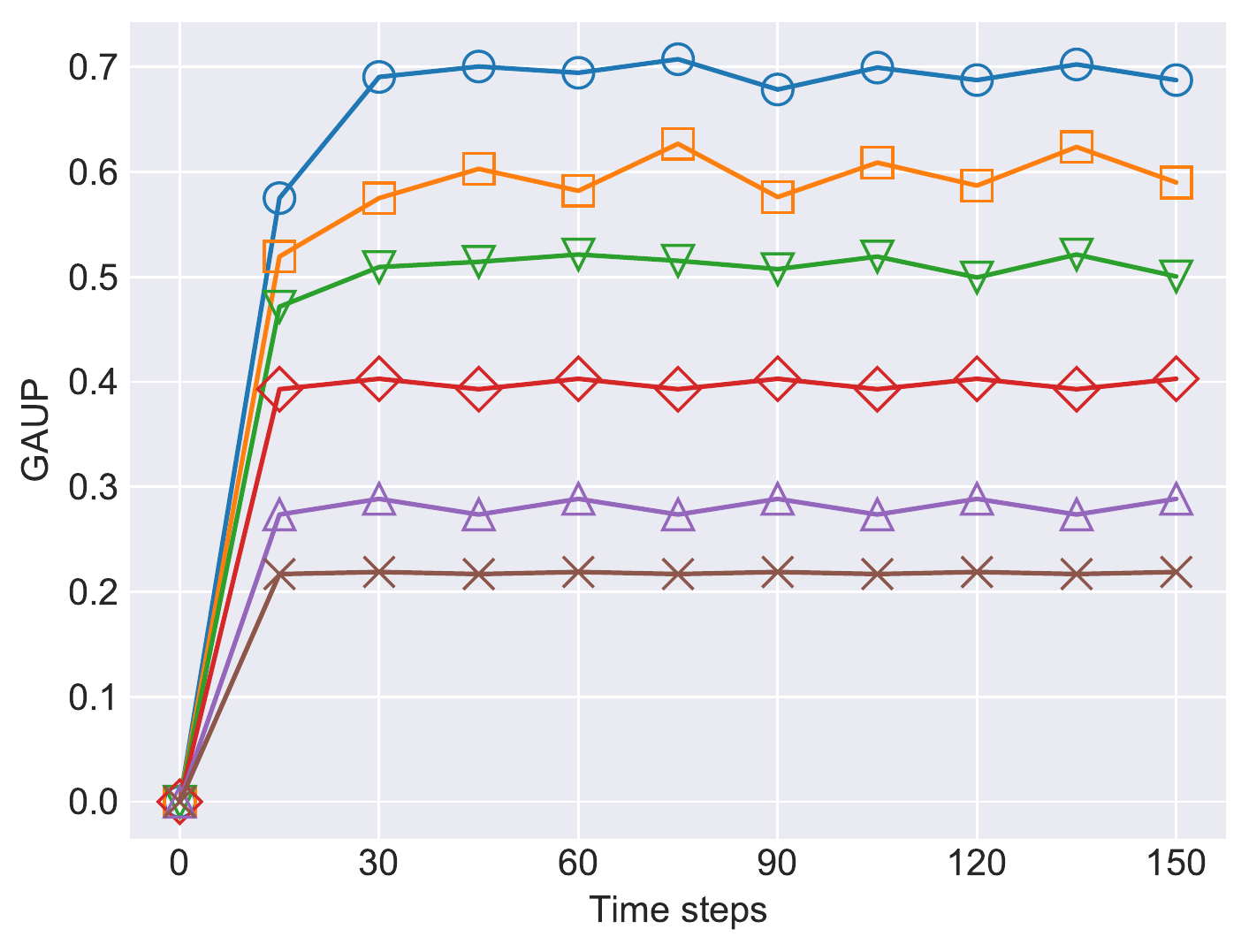}\label{fig:gaup-d}}\\
  \subfigure{\includegraphics[width=\textwidth]{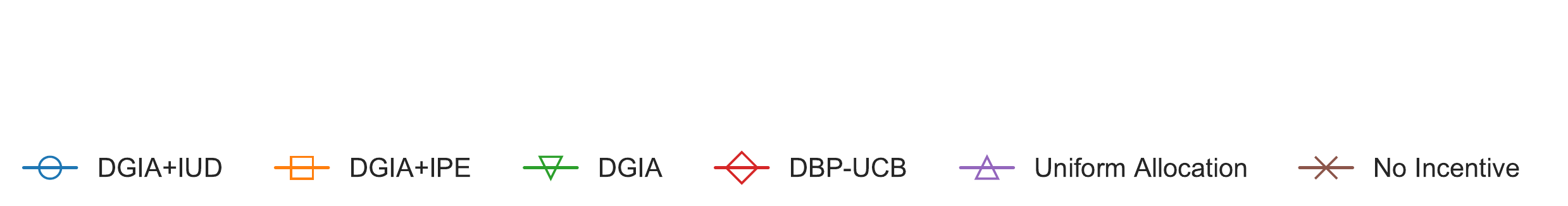}}
  \caption{Comparison of GAUP}
  \label{fig:gaup}
\end{figure*}

We first evaluate the performance by using GAUP. As we can observe from Figure \ref{fig:gaup}, IUD+DGIA can produce higher GAUP in all four networks. Meanwhile, IUD+DGIA can converge fast in a short time, and keep GUAP relatively stable. The gap between IUD+DGIA and DGIA implies that IUD is effective in engaging influential users, which effectively improves the overall performance. By contrast, the performance of IPE+DGIA is not always satisfying. Although IPE+DGIA can reach higher GAUP than the other four approaches in Facebook and Email networks, it produces similar results as that of DGIA in Twitter and Wiki networks. The reason is that IPE fails to discover correct influential users, and incorrectly regards non-influential users as influential users. 

Furthermore, among approaches that are not influence-aware, DGIA is superior to the other three approaches in Twitter and Email networks. However, the performance of DBP-UCB is close to DGIA in the Facebook network and even exceeds it in the Wiki network. This is due to the reason that DBP-UCB is a MAB-based approach that requires a sufficient budget to explore the optimal value of incentives, and the budget provided in Twitter and Email networks is insufficient for that. However, DGIA explores the proper value of incentives by analyzing users' incentive sensitivity, and also restricts incentives based on GAUP. This feature causes that DGIA is able to perform well when the budget is limited, but also limits the performance even the budget is very sufficient. Uniform approach can only improve limited GAUP due to the limited budget. Also, the performance of No Incentive approach demonstrates that only a few users would select $a^*$ without receiving incentives in four networks. The tendencies of these two approaches indicate that the necessity of deploying an adaptive incentive allocation approach. On the other hand, as we can see from Figures \ref{fig:gaup-a} and \ref{fig:gaup-b}, the tendencies of Uniform approach and No Incentive approach both rise at the beginning then fall down to a steady phase. It is due to that some influential users select $a^*$ and influence their neighbors to select $a^*$ in the beginning. However, some of these influential users are also influenced by other users who choose actions other than $a^*$. This phenomenon indicates that social influence plays a crucial role in affecting users' decision-making in the social network.

\begin{figure*}[!t]
  \centering
  \subfigure[Facebook, $B_t=200$]{\includegraphics[width=0.49\textwidth]{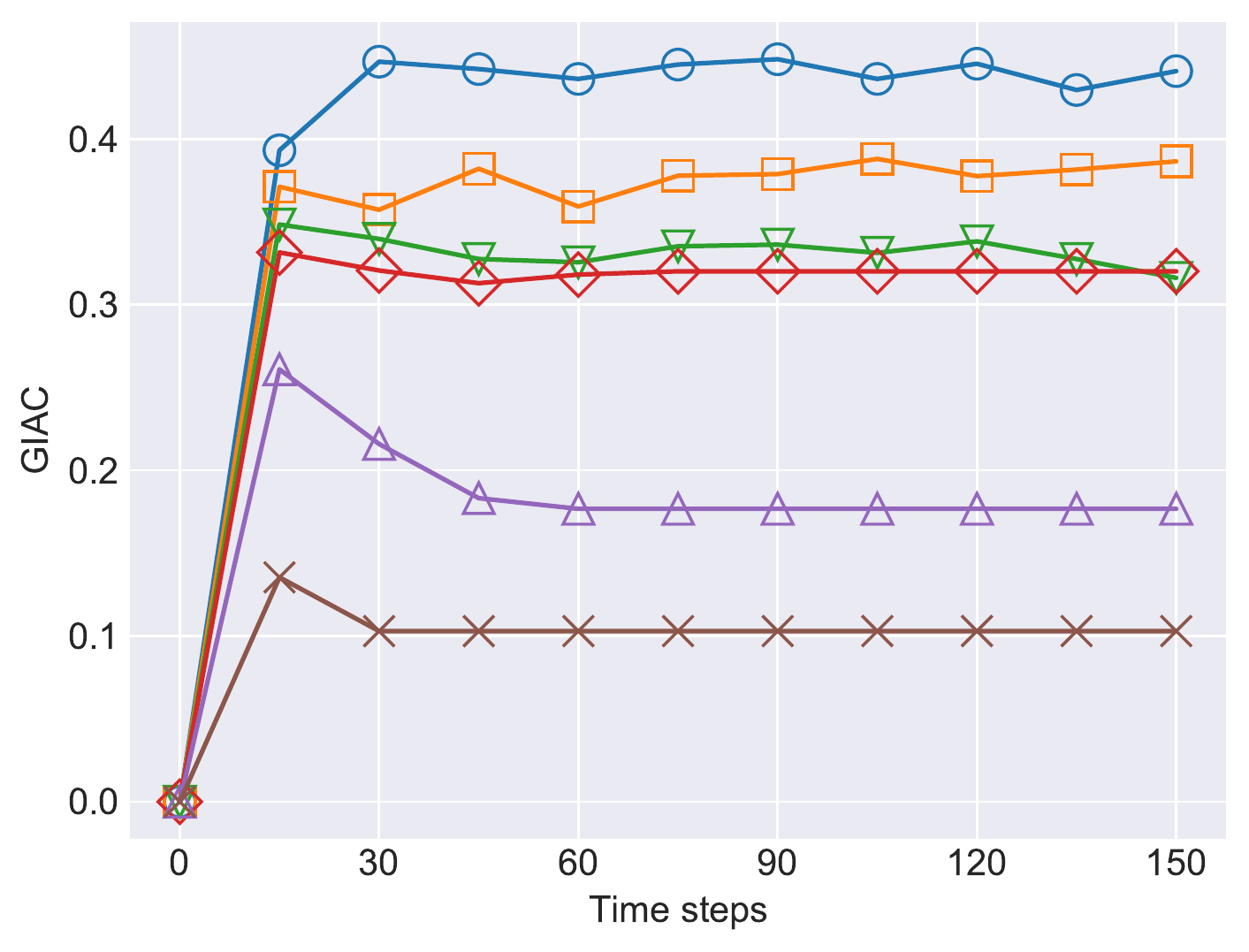}}
  \subfigure[Twitter, $B_t=20$]{\includegraphics[width=0.49\textwidth]{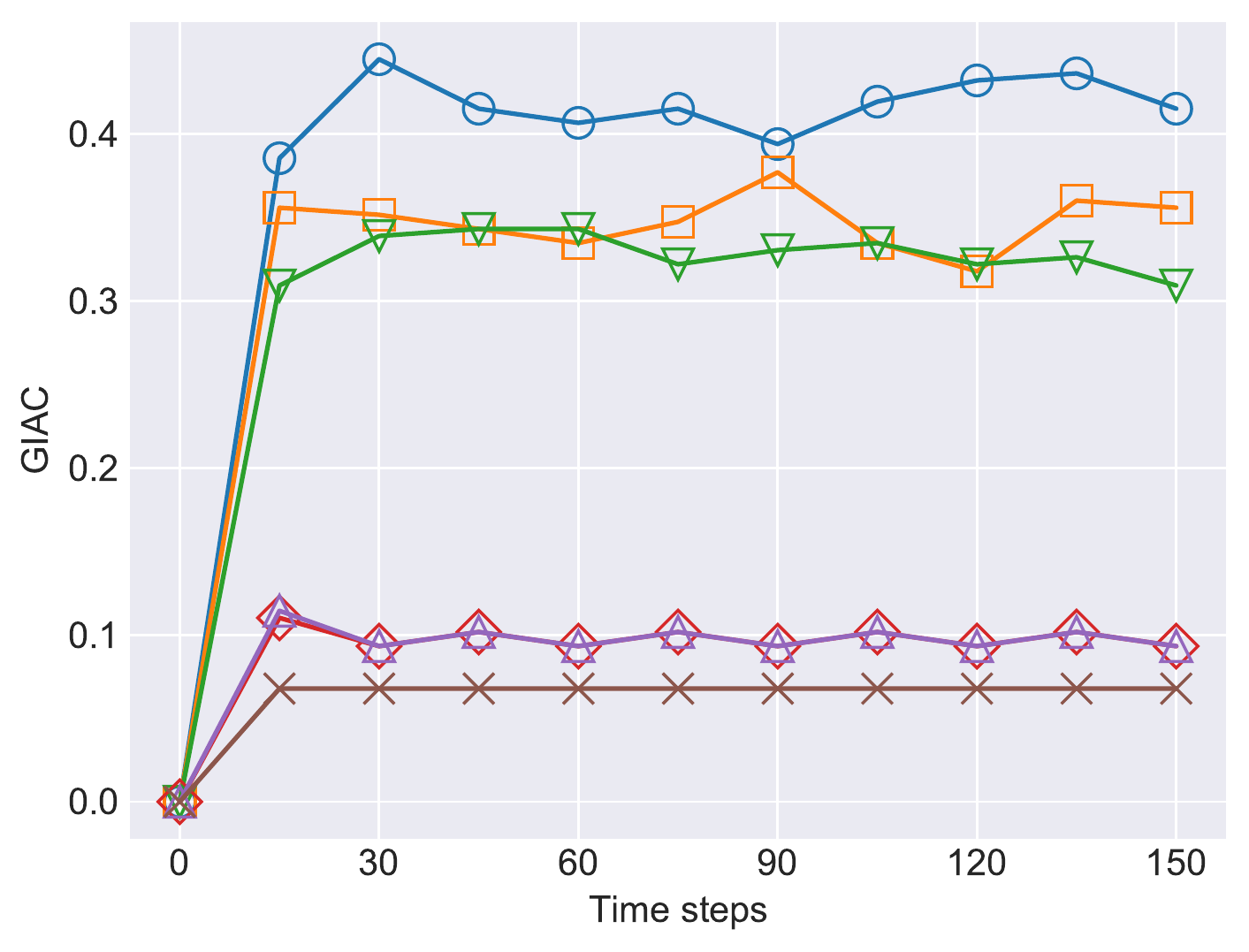}}
  \subfigure[Wiki, $B_t=700$]{\includegraphics[width=0.49\textwidth]{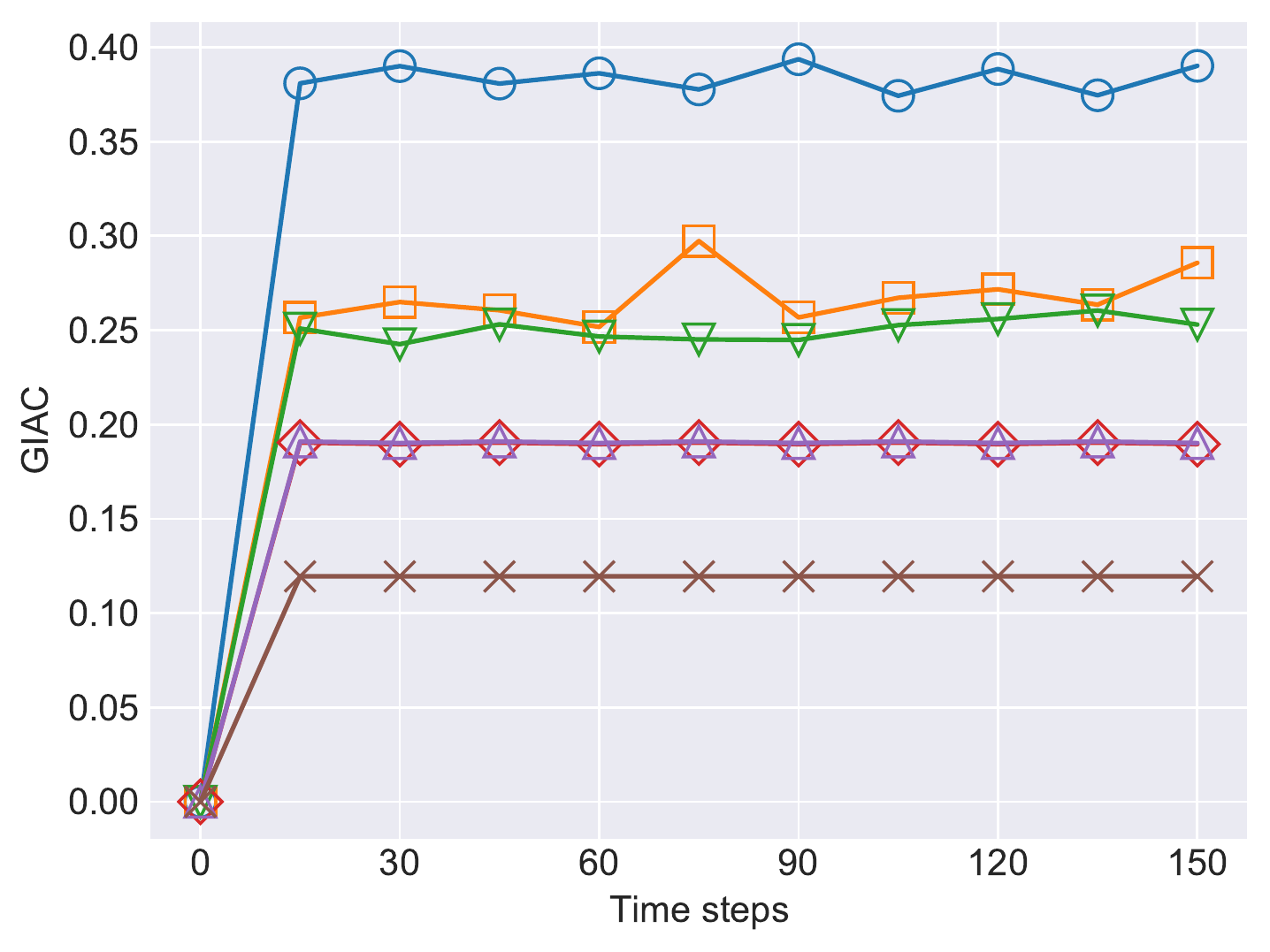}}
  \subfigure[Email, $B_t=50$]{\includegraphics[width=0.49\textwidth]{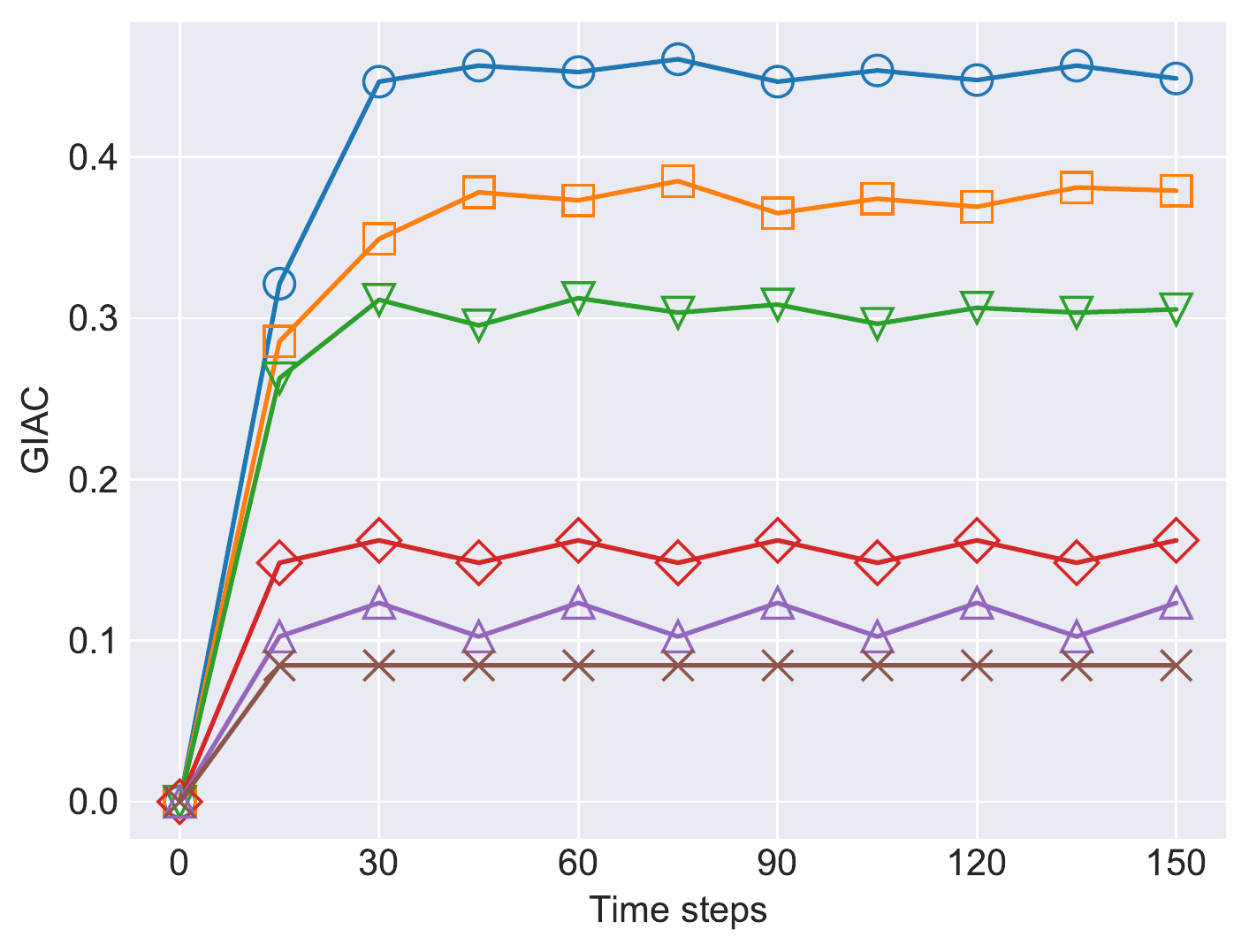}}\\
  \subfigure{\includegraphics[width=\textwidth]{Figures/Exp/legend.pdf}}
  \caption{Comparison of GIAC}
  \label{fig:giac}
\end{figure*}

\subsubsection{Performance on GIAC}
Based on the results shown in Figure \ref{fig:giac}, we can observe that IUD+DGIA can produce higher GIAC than the other five approaches in all four networks. The evident gaps between IUD+DGIA and the other five approaches reflect that IUD can effectively incentivize users by utilizing social influence. IPE+DGIA performs well in Facebook and Email networks, but gives similar performance to DGIA in Twitter and Wiki networks. This is due to the reason that IPE only estimates users' influential degrees based on their latest behaviors which easily leads to incorrect estimation. On the other hand, while DGIA performs worse than IUD+DGIA and IPE+DGIA, its performance is still better than DBP-UCB, Uniform, and No Incentive approaches. The reason behind that is the mechanism of DGIA which limits incentives provided to users makes a number of users are incentivized due to social influence rather than sufficient incentives. 

Another interesting finding from the results demonstrated in Figures \ref{fig:gaup} and \ref{fig:giac} is relating to the relationship between GAUP and GIAC. Since GAUP represents the percentage of incentivized users and GIAC represents the percentage of users who are incentivized by social influence, the difference between these two metrics can reflect the percentage of users who receive sufficient incentives at each time step. Take the Facebook network as an example, the GAUP and the GIAC of IUD+DGIA approximate 70\% and 45\%, respectively, and the difference of that is 25\%. Similarly, DBP-UCB gives 55\% on GAUP and 33\% on GIAC, and the difference of which is 22\%. Furthermore, by comparing the GAUP and GIAC of the other approaches, we notice that they all have a similar difference in the same network. While the results of these differences are close, these approaches still produce different GAUP and GIAC eventually. Apparently, the approach which produces higher GIAC can produce higher GAUP, since properly incentivizing influential users helps to affect more users with fewer incentives. This phenomenon also confirms the effectiveness of the proposed approach.

\begin{table*}[!t]
\centering
\scalebox{0.9}{
\resizebox{\textwidth}{!}{
\begin{tabular}{ccccccccc}
\hline
Dataset                   & \begin{tabular}[c]{@{}c@{}}Total\\ Budget\end{tabular}                      & Approach     & \begin{tabular}[c]{@{}c@{}}Spending\\ Budget\end{tabular}    & $\overline{\mu}$         & $\overline{\tau}$           & Utilization    & $R(\overline{\mu})$&$R(\overline{\tau})$\\\hline
\multirow{6}{*}{Facebook} & \multirow{6}{*}{30,000.00}  & IUD+DGIA     & 20,132.795          & \textbf{0.665} & \textbf{0.419} & 0.671          & \textbf{0.698}   & \textbf{0.461}   \\
                          &                             & IPE+DGIA     & 16,935.677          & 0.583          & 0.364          & 0.565          & 0.685            & 0.452            \\
                          &                             & DGIA         & 14,213.218          & 0.519          & 0.321          & 0.474          & 0.680            & 0.447            \\
                          &                             & Uniform      & \textbf{9,784.353}  & 0.324          & 0.188          & \textbf{0.326} & 0.391            & 0.241            \\
                          &                             & DBP-UCB      & 29,844.800          & 0.544          & 0.312          & 0.995          & 0.350            & 0.204            \\
                          &                             & No Incentive & /          & 0.196          & 0.109          & /              & /                & /                \\\hline
\multirow{6}{*}{Twitter}  & \multirow{6}{*}{3,000.00}   & IUD+DGIA     & 1,184.260           & \textbf{0.691} & \textbf{0.398} & 0.395          & 1.313            & 0.839            \\
                          &                             & IPE+DGIA     & 1,099.433           & 0.596          & 0.327          & 0.366          & 1.157            & 0.710            \\
                          &                             & DGIA         & 843.381             & 0.572          & 0.316          & 0.281          & \textbf{1.421}   & \textbf{0.885}   \\
                          &                             & Uniform      & \textbf{830.932}    & 0.275          & 0.099          & \textbf{0.277} & 0.372            & 0.115            \\
                          &                             & DBP-UCB      & 1,111.750           & 0.312          & 0.100          & 0.371          & 0.376            & 0.088            \\
                          &                             & No Incentive & /           & 0.172          & 0.067          & /              & /                & /                \\\hline
\multirow{6}{*}{Wiki}     & \multirow{6}{*}{105,000.00} & IUD+DGIA     & 41,255.169          & \textbf{0.646} & \textbf{0.370} & 0.393          & \textbf{0.788}   & \textbf{0.499}   \\
                          &                             & IPE+DGIA     & 42,715.642          & 0.511          & 0.260          & 0.407          & 0.427            & 0.210            \\
                          &                             & DGIA         & \textbf{28,806.743} & 0.468          & 0.245          & \textbf{0.274} & 0.479            & 0.259            \\
                          &                             & Uniform      & 46,152.579          & 0.437          & 0.188          & 0.440          & 0.227            & 0.030            \\
                          &                             & DBP-UCB      & 47,013.200          & 0.438          & 0.187          & 0.448          & 0.225            & 0.029            \\
                          &                             & No Incentive & /         & 0.337          & 0.174          & /              & /                & /                \\\hline
\multirow{6}{*}{Email}    & \multirow{6}{*}{7,500.00}   & IUD+DGIA     & 4,050.918           & \textbf{0.662} & \textbf{0.419} & 0.540          & \textbf{0.824}   & \textbf{0.621}   \\
                          &                             & IPE+DGIA     & 4,053.888           & 0.573          & 0.348          & 0.541          & 0.660            & 0.490            \\
                          &                             & DGIA         & 3,075.611           & 0.497          & 0.290          & 0.410          & 0.685            & 0.503            \\
                          &                             & Uniform      & \textbf{2,103.085}  & 0.278          & 0.111          & \textbf{0.280} & 0.222            & 0.098            \\
                          &                             & DBP-UCB      & 5,972.800           & 0.394          & 0.152          & 0.796          & 0.223            & 0.087            \\
                          &                             & No Incentive & /           & 0.216          & 0.083          & /              & /                & /               \\\hline
\end{tabular}%
}
}
\caption{Comparison of Efficiency of Budget Use}
\label{tab:budget}
\end{table*}

\subsubsection{Performance on Efficiency of Budget Use}
With the same constraint of the budget, different approaches may spend different amounts of the budget. Hence, we also attempt to understand the efficiency of budget use. Five metrics are adopted for evaluating each approach, i.e., average DGIA $\overline{\mu}$, average GIAC $\overline{\tau}$, utilization of budget, the return rate of DGIA $R(\mu)$, and the return rate of GIAC $R(\tau)$. $R(\mu)$ and $R(\tau)$ are formulated by using Equations \ref{equ:r_mu} and \ref{equ:r_nu}, where $\overline{\mu}'$ and $\overline{\tau}'$ denote the average GAUP and GIAC of No Incentive approach, respectively.

\begin{equation}\label{equ:r_mu}
R(\overline{\mu}) = \frac{\overline{\mu} - \overline{\mu}'}{utilization}
\end{equation}
\begin{equation}\label{equ:r_nu}
R(\overline{\tau}) = \frac{\overline{\tau} - \overline{\tau}'}{utilization}
\end{equation}

As we can observe from Table \ref{tab:budget}, IUD+DGIA outperforms the other five approaches in all four networks with respect of $\overline{\mu}$ and $\overline{\tau}$. While IUD+DGIA spends more budget than IPE+DGIA and DGIA, $R(\overline{\mu})$ and $R(\overline{\tau})$ demonstrate that IUD+DGIA is efficient in the use of the budget. This is due to the reason that IUD requires more budget for analyzing influential users in the initial phase. However, after identifying influential users, IUD+DGIA can incentivize users with considering social influence, so that the budget is saved. By contrast, Uniform approach spends the least budget on Facebook, Twitter, and Email networks, since it fails to incentivize most users. Meanwhile, DBP-UCB spends more budget but produces limited effects since it provides excessive incentives to some users.

\begin{table*}[!t]
\centering
\scalebox{1}{
\resizebox{\textwidth}{!}{%
\begin{tabular}{ccccc}
\hline
Networks & Initial Number of Users & Final Number of Users & \begin{tabular}[c]{@{}c@{}}The number of users \\ joining the network \end{tabular} & \begin{tabular}[c]{@{}c@{}}The number of users \\ leaving the network\end{tabular} \\ \hline
DN1      & 1446                & 3075                 & {[}1,50{]}                                                                                       & {[}1,20{]}                                                                                                \\
DN2      & 1446              & 2936                & {[}1,50{]}                                                                                       & {[}1,50{]}                                                                                                \\
DN3      & 1446               & 3208                  & {[}1,100{]}                                                                                      & {[}1,50{]}                                                                                                \\ \hline
\end{tabular}%
}
}
\caption{Statistics of Dynamic Networks}
\label{tab:dynamic}
\end{table*}

\begin{figure*}[!t]
  \centering
  \subfigure[GAUP of DN1]{\includegraphics[width=0.32\textwidth]{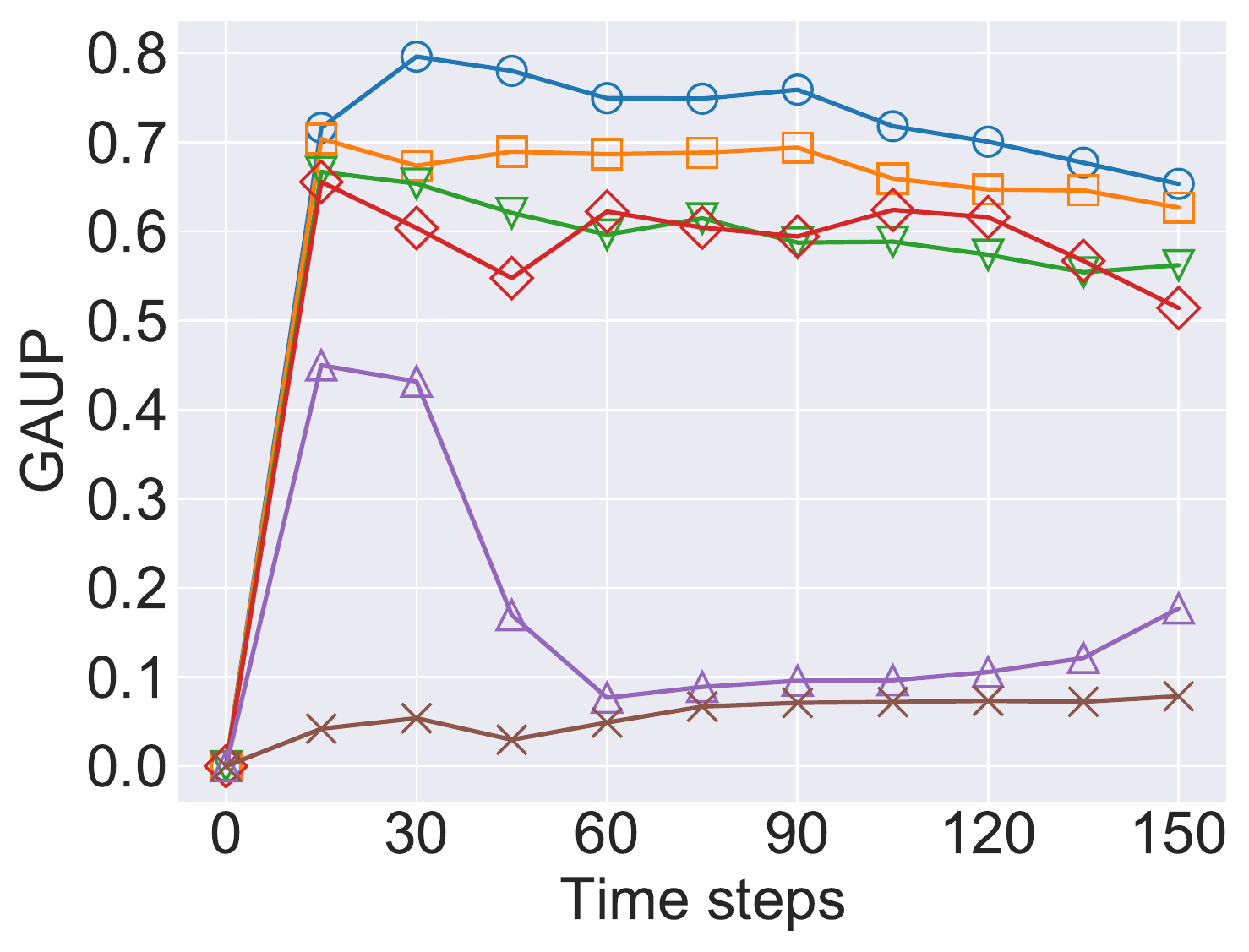}}
  \subfigure[GIAC of DN1]{\includegraphics[width=0.32\textwidth]{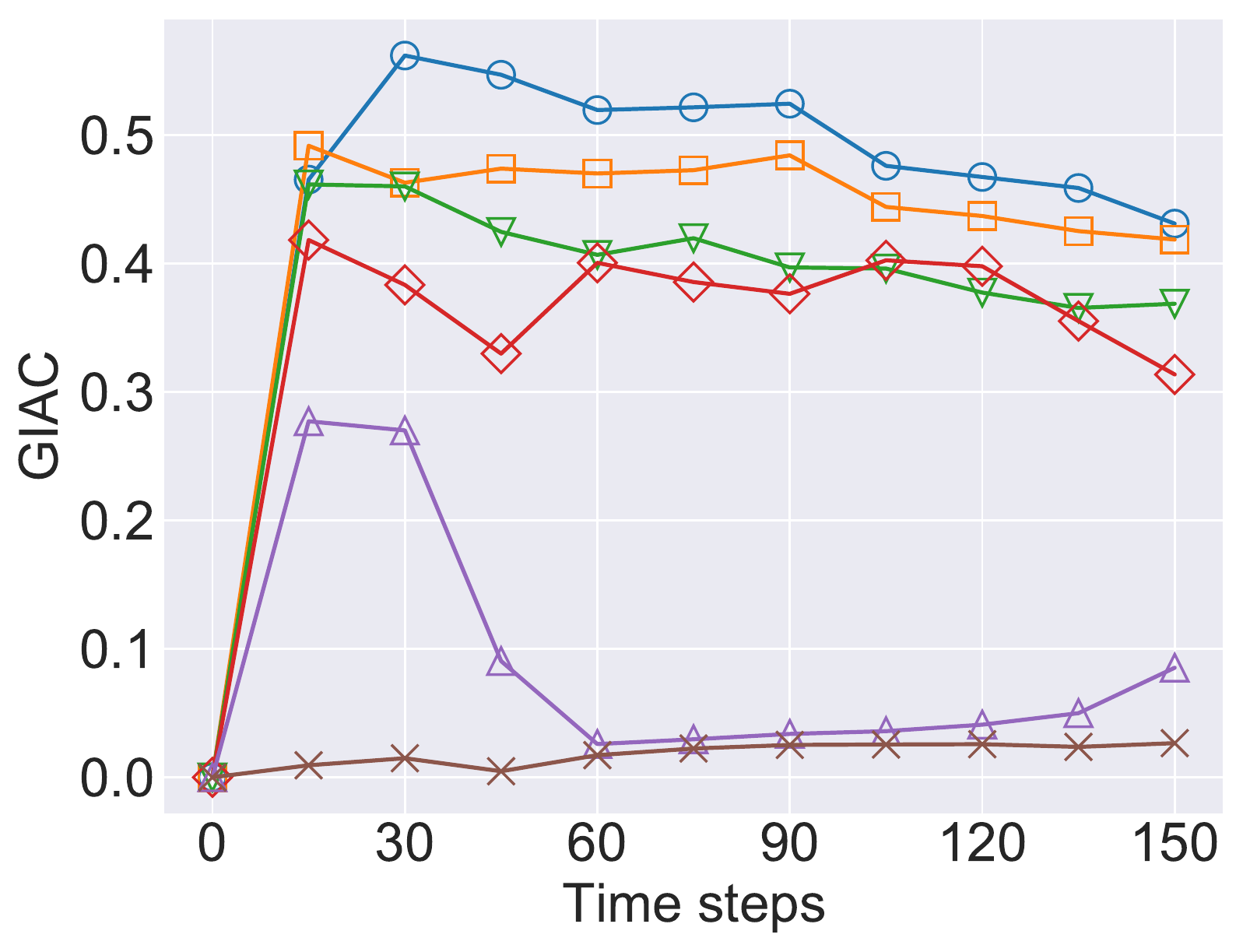}}
  \subfigure[GAUP of DN2]{\includegraphics[width=0.32\textwidth]{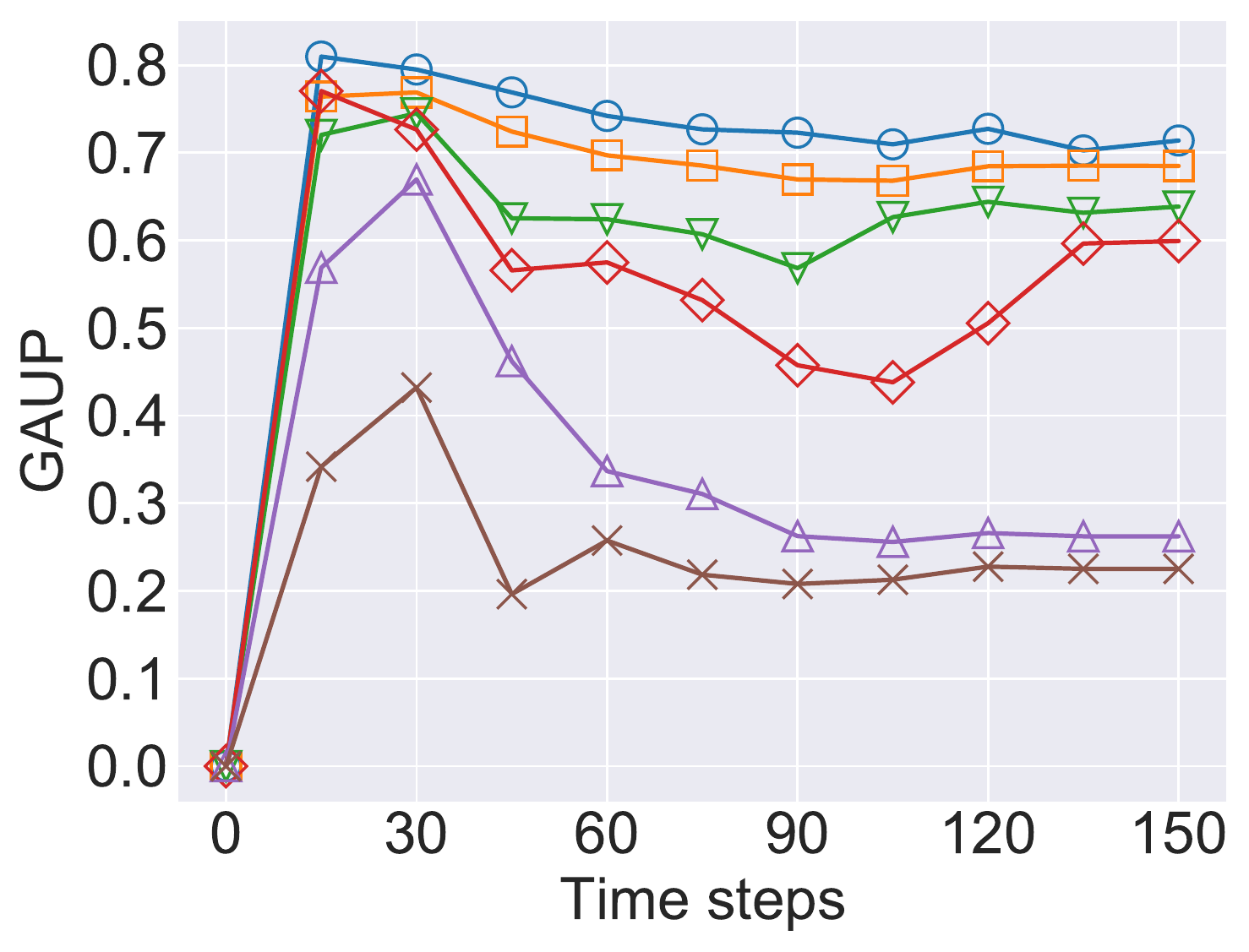}}
  \subfigure[GIAC of DN2]{\includegraphics[width=0.32\textwidth]{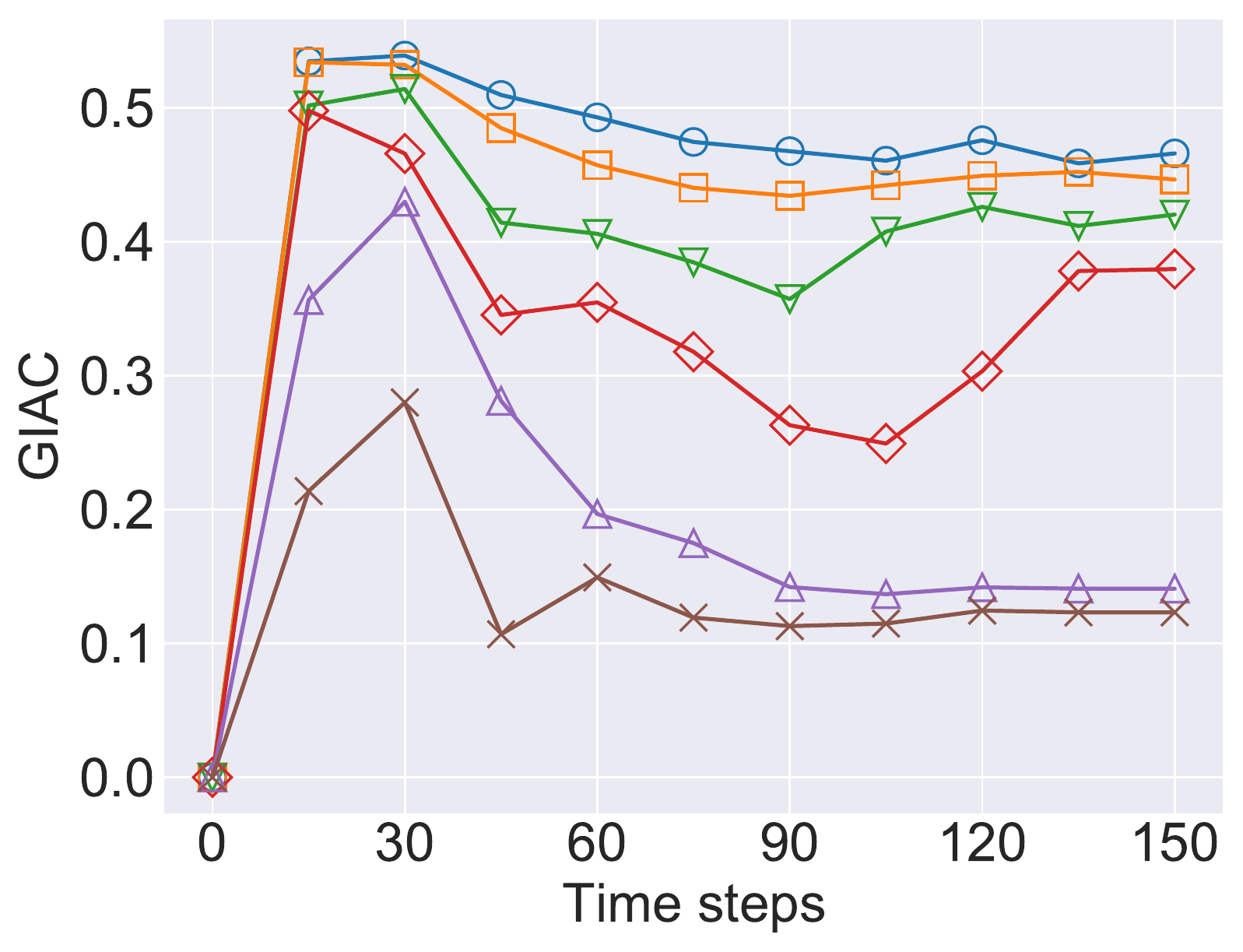}}
  \subfigure[GAUP of DN3]{\includegraphics[width=0.32\textwidth]{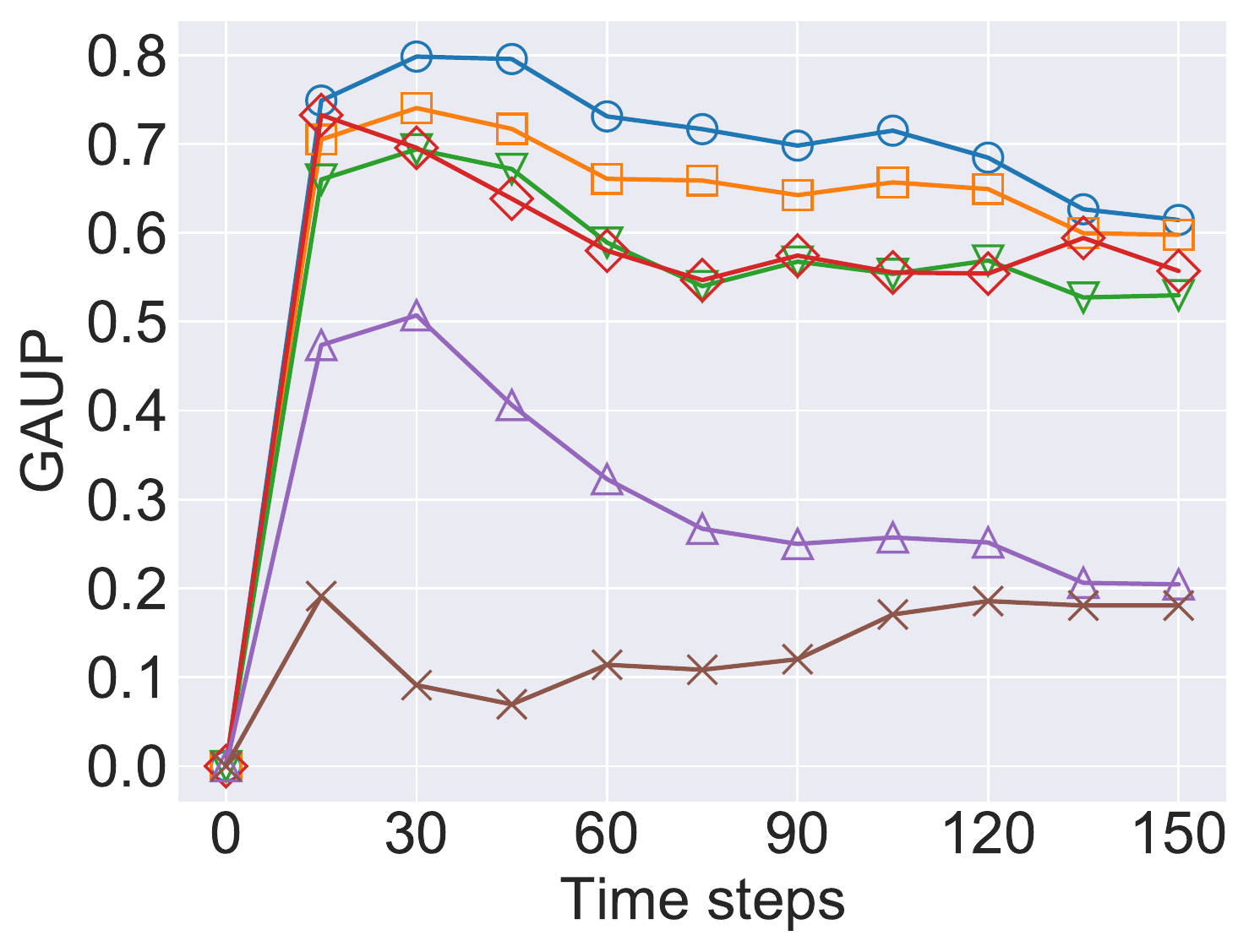}}
  \subfigure[GIAC of DN3]{\includegraphics[width=0.32\textwidth]{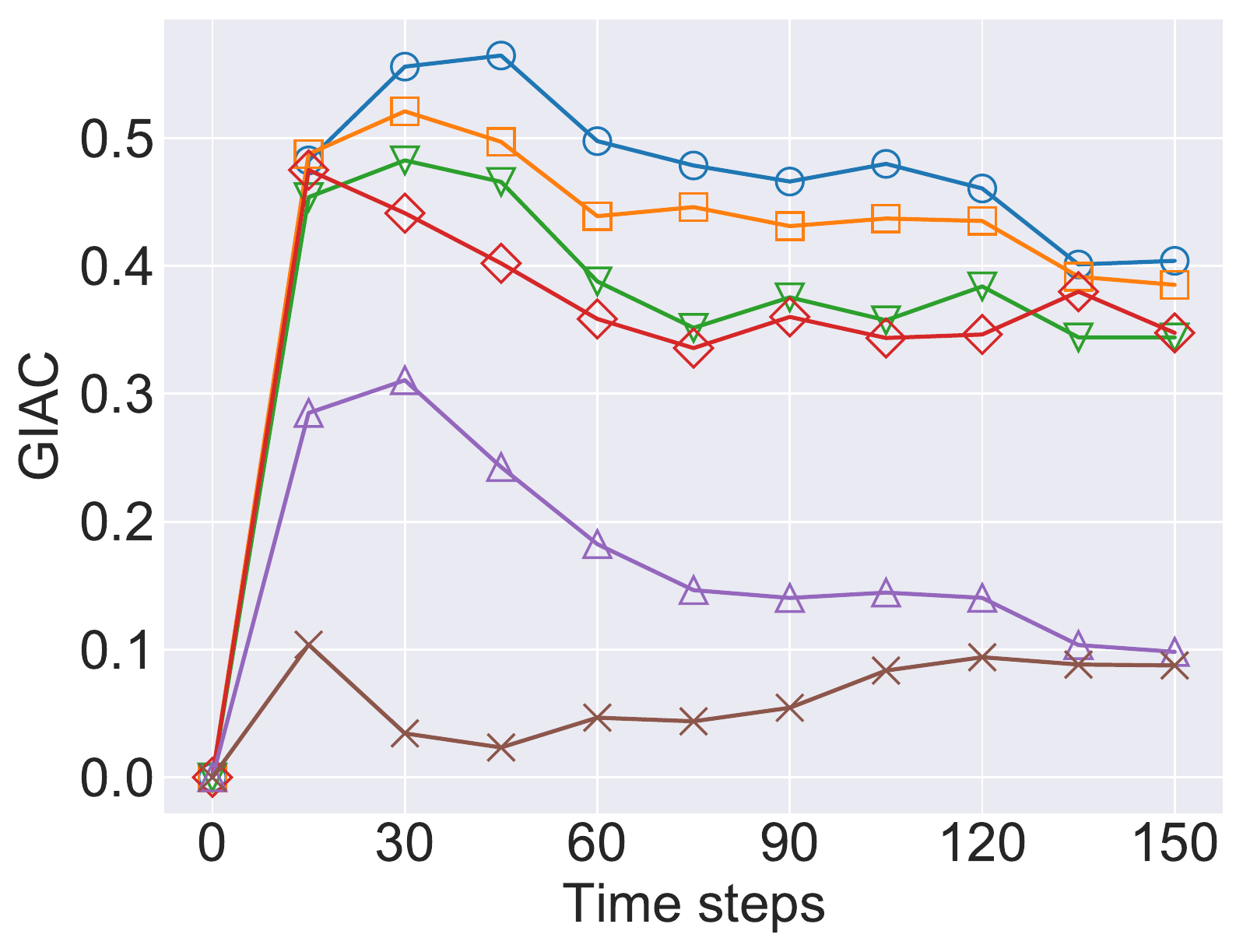}}
  \subfigure{\includegraphics[width=\textwidth]{Figures/Exp/legend.pdf}}
  \caption{Performance of Approaches on Three Dynamic Networks}
  \label{fig:dy}
\end{figure*}

\subsection{Experimental Results on Dynamic Networks}\label{subsec:dynamic}
In the real world, the structure of a social network is possibly dynamic since new users may join the network whereas existing users may leave as well. Hence, to investigate the performance of the proposed approach in dynamic networks, we build three synthetic networks based on the Haverford dataset. The topology of the network would change at the beginning of each time step. We set that a random number of existing users would leave the network at every five time steps, and they cannot influence users in the network anymore after leaving. Meanwhile, a random number of new users would join the network at every time step, and the number of users that a new user can connect to is from 1 to 20. The budget provided for incentive allocation is 70 per time step. The statistics of these three dynamic networks are listed in Table \ref{tab:dynamic}.

Figure \ref{fig:dy} demonstrates the performance of the approaches in dynamic networks. Overall, IUD+DGIA still outperforms the other five approaches by producing higher GAUP and GIAC. Meanwhile, IPE+DGIA and DGIA both bring close results to that of IUD+DGIA. Besides, the GAUP of DBP-UCB is similar to that of DGIA in DN1 and DN2. Whereas in DN3, we notice that all DGIA based approaches and DBP-UCB perform similarly in the final phase. Based on the performance of Uniform approach in DN3, a possible explanation is that most influential users leave the network one after another. Meanwhile, the variation of Uniform approach in DN1 and DN2 is also due to the frequent variety of topology, since we have to spend more budget on incentivizing users directly if influential users leave the network.

\subsection{Discussion}
In the experiments, we simulated a group of users' decision-making in a social network by deploying the ADM. Under four real social networks, we evaluated the effectiveness of the proposed ensemble approach of the IUD and the DGIA by using two major metrics, i.e., GAUP and GIAC. Meanwhile, we evaluated the budget efficiency of the proposed approach. Furthermore, experiments on dynamic social networks are conducted to assess the performance of the proposed approach. 

The experimental results show that the proposed approaches can perform better in incentive allocation. Given a limited budget and time span, IUD+DGIA can effectively and efficiently incentivize most users in a social network. The comparison of three DGIA based approaches (i.e., IUD+DGIA, IPE+DGIA, and DGIA) verifies that 1) IUD can identify influential users more accurately, and 2) Allocating incentives with considering influential ability is helpful to incentivize more users in a social network. Whereas the comparison among DGIA, DBP-UCB, and Uniform approach also demonstrates the effectiveness of DGIA in incentive allocation. The insights uncovered from the experiments can be summarized as follows:
\begin{itemize}
\item The results from Experiment 1 appear consistent with Theorem \ref{theorem:1}, demonstrating that given a fixed budget, incentivizing influential users to propagate influence can help save budget when incentivizing non-influential users.
\item Given a limited budget, leveraging social influence to indirectly incentivize users' behavior is a better way than directly incentivize users in social networks.
\item Even though estimating users' influential ability based on their historical behaviors might not be accurate, IUD can be effective in assisting incentive allocation when the information of social networks is unknown.
\item Experiment 2 verifies that the proposed approaches can perform well in dynamic networks. However, the performance of proposed approaches might fall down if members in the network change frequently, since the proposed approaches require time to better estimate users' influential degree and incentive sensitivity.
\end{itemize}

\section{Conclusions}\label{sec:conclusions}
In this paper, we investigated the incentive allocation problem in unknown social networks. To tackle this problem, we first proposed the ADM model to describe the process of users' decision-making under the incentive and social influence. We also proposed a novel algorithm, i.e., IUD,  to estimate the influential probability between users, for discovering influential users in an unknown network. Subsequently, a dynamic-game-based incentive allocation algorithm, i.e., DGIA, to analyze the probability that a user accepts the incentive and determine the value of incentives, is also proposed. The experimental results demonstrate that the ensemble approach of IUD and DGIA can outperform other approaches both in static and dynamic networks. Meanwhile, the comparison of GIAC demonstrates that the IUD algorithm can effectively discover influential users, and the comparison of budget efficiency demonstrates DGIA algorithm can effectively incentivize users with a limited budget. 

In the future, we will continue to explore potential factors for affecting users' behaviors, and continue to investigate approaches for influential users discovery and more effective incentive allocation in unknown networks.



\bibliographystyle{elsarticle-num-names-sort}
\bibliography{reference}
\end{document}